\pdfoutput = 1
\documentclass[11pt]{article}
\usepackage{url}
\usepackage{color}
\usepackage{graphicx}
\usepackage{amsfonts}
\usepackage{amsmath}
\usepackage{amssymb}
\usepackage{latexsym}
\usepackage{rotating}

\usepackage{colortbl}
\usepackage[table]{xcolor}
\usepackage{multirow}
\newcommand\open{\textcolor{blue}{{\bf \sc Open}}}

\definecolor{cyell}{rgb}{.9,0,9,0}

\definecolor{kgray}{rgb}{224,224,224}

\newcommand{\hide}[1]{}

\newcommand\cone{{C}}
\newcommand\red[1]{\textcolor{red}{#1}}

\newcommand\Pt{\mathcal P}

\newcommand\pyy{p_{{\tt YY}}}
\newcommand\ang[1]{\widehat{#1}}

\newcommand{\arr}[1]{\overrightarrow{#1}}

\newtheorem{theorem}{{\bf Theorem}}

\newtheorem{lemma}[theorem]{Lemma}

\newcommand{\qed}{\rule{0.5em}{1.5ex}}
\newcommand{\fqed}{{\hfill~\qed}}
\newenvironment{proof}{{\noindent \bf Proof.}}
                      {{\hfill \fqed} \vspace{1em}}

\newcommand{\eproof}{{\hfill~\fqed} \vspace{1em}}
\newcommand{\notyet}[1]{}

\begin{document}

\title{An Infinite Class of Sparse-Yao Spanners}
\author{Matthew Bauer
    \thanks{ Department of Computing Sciences, Villanova University, Villanova, PA, USA. \protect\url{mbauer03@villanova.edu}.}
\and
Mirela Damian
    \thanks{ Department of Computing Sciences, Villanova University, Villanova, PA, USA. \protect\url{mirela.damian@villanova.edu}.}
}

\date{}

\maketitle

\begin{abstract}
We show that, for any integer $k \ge 6$, the Sparse-Yao graph $YY_{6k}$ (also known as Yao-Yao) is a spanner with stretch factor $11.67$. The stretch factor drops down to $4.75$ for $k \ge 8$.
\end{abstract}

\section{Introduction}
Let $\Pt$ be a finite set of points in the plane. The Yao graph and the Theta graph for $\Pt$ are directed geometric graphs with vertex set $\Pt$ and directed edges defined by an integer parameter $k \ge 2$ as follows. Fix a coordinate system and consider the rays obtained by a counterclockwise rotation of the positive $x$-axis about the origin by angles of $2j\pi/k$, for integer $0 \le j \le k-1$. Each pair of successive rays defines a cone whose apex is the origin, for a total of $k$ cones. Translate these cones to each node $a \in \Pt$, then connect $a$ to a ``nearest neighbor'' in each of the $k$ cones using directed edges rooted at $a$. This yields an out-degree of at most $k$. The Yao and Theta graphs differ in the way the ``nearest neighbor'' is defined.
In the case of Yao graphs, the neighbor nearest to $a$ in a cone $\cone$ is a point $b\neq a$ that lies in $C$ and minimizes the Euclidean distance $|ab|$ between $a$ and $b$; ties are broken arbitrarily.
In the case of Theta graphs, the neighbor nearest to $a$ in a cone $\cone$ is a point $b \neq a$ that lies in $C$ and minimizes the Euclidean distance between $a$ and the orthogonal projection of $b$ onto the bisector of $C$; ties are broken in favor of a neighbor $b$ that minimizes $|ab|$, and in the case of two such neighbors, one is arbitrarily selected.
Henceforth, we will refer to the Yao graph as $Y_k$ and the Theta graph as $\Theta_k$.
%
\begin{figure}[htpb]
\centering
\includegraphics[width=0.8\linewidth]{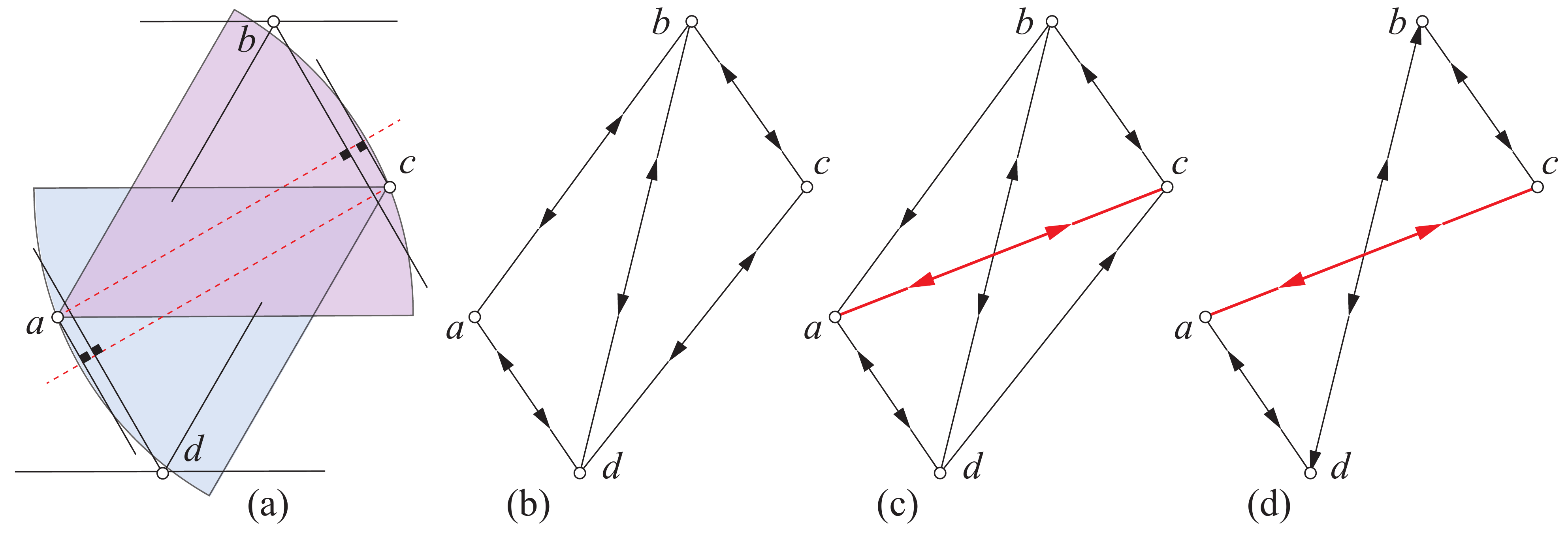}
\caption{(a) $\Pt = \{a,b,c,d\}$ (b) Theta graph $\Theta_6$ (c) Yao graph $Y_6$ (d) Sparse-Yao graph $YY_6$}
\label{fig:yaothetaex}
\end{figure}
%
Fig.~\ref{fig:yaothetaex} shows a simple example with four points $\Pt = \{a,b,c,d\}$.
The nonempty cones at each point, to be used in constructing $\Theta_6$ and $Y_6$ (for fixed $k = 6$),
are delineated in Fig.~\ref{fig:yaothetaex}a. In the cone $\cone$ with apex $a$ containing both points
$b$ and $c$, note that $|ac| < |ab|$ (because $b$ lies strictly outside the circle centered at
$a$ of radius $|ac|$), but the orthogonal projection of $b$ on the bisector of $\cone$ is closer to $a$
compared to the orthogonal projection of $c$ on the bisector of $C$. Consequently, $\Theta_6$ selects
$\arr{ab}$ in $\cone$ (see Fig.~\ref{fig:yaothetaex}b), whereas $Y_6$ selects $\arr{ac}$ in $\cone$ (see Fig.~\ref{fig:yaothetaex}c).
Similarly, $\Theta_6$ favors $\arr{cd}$ over $\arr{ca}$ in the cone with apex $c$ containing both $a$ and $d$,
whereas $Y_6$ selects $\arr{ca}$.

Interest in Theta graphs and Yao graphs has increased with the advancement of wireless network technologies and the need for efficient communication. Among other properties, communication graphs are required to include short paths between any pair of nodes to enable efficient routing, and to have low degree to reduce MAC-level contention and interference~\cite{DegreeMac06}. It turns out that both $\Theta_k$ and $Y_k$ obey the first requirement (for any $k > 6$ and other specific values of $k$, as detailed in Table~\ref{tab:yao}), but fail to satisfy the second one. Imagine for example the simple scenario in which $\Pt$ consists of $n-1$ nodes placed on the circumference of a circle with center node $a$. Then, for $k \ge 6$, each of $\Theta_k$ and $Y_k$ will have an edge directed from each of the $n-1$ nodes towards $a$, because $a$ is ``nearest'' in one of their cones. So each of $\Theta_k$ and $Y_k$ has out-degree $k$, but in-degree $n-1$. To overcome the problem of potential high in-degree at a node, the Sparse-Yao graph $YY_k$, also known as the Yao-Yao graph, has been introduced. The graph $YY_k \subseteq Y_k$ is obtained by applying a second Yao step to the set of incoming Yao edges: for each node $a$ and each cone rooted at $a$ containing two or more incoming edges, retain a shortest incoming edge and discard the rest; ties are broken arbitrarily. Fig.~\ref{fig:yaothetaex}d shows that Sparse-Yao graph $YY_6$ corresponding to the Yao graph $Y_6$ from Fig.~\ref{fig:yaothetaex}c: because $\arr{ba}$ and $\arr{ca}$ lie in one same cone with apex $a$, and because $|ca| < |ba|$, $YY_6$ keeps $\arr{ca}$ and discards $\arr{ba}$; similarly, because $|ac| < |dc|$, $YY_6$ keeps $\arr{ac}$ and discards $\arr{dc}$.
The degree of $YY_k$ is bounded above by $2k$ (at most $k$ outgoing edges and at most $k$ incoming edges at each node). Although not as popular, the Sparse-Theta graph can be defined analogously.

Ignore for the moment the direction of the edges in $\Theta_k$, $Y_k$ and $YY_k$, and view these graphs as \emph{undirected} graphs.
We present some interesting properties of these graphs, along with our main result, after a few brief definitions. Let $G$ be an undirected graph with vertex set $\Pt$. The \emph{length} of a path in $G$ is the sum of the Euclidean lengths of its constituent edges. For a fixed real $t \ge 1$, we say that $G$ is a $t$-\emph{spanner} for $\Pt$ if, for each pair of points $a, b \in \Pt$, there is a path in $G$ whose length is at most $t|ab|$; the value $t$ is called the \emph{stretch factor} of $G$. The graphs $\Theta_k$ and $Y_k$ are known to be spanners for any $k \ge 6$; the stretch factors for specific ranges of $k$ are listed in Table~\ref{tab:yao}. Very little in comparison is known about Sparse-Yao graphs. The only existing results are negative and show that $YY_k$, for $k \in \{2,3,4,6\}$, are not $t$-spanners for any constant real value $t$. For a comprehensive discussion of spanners, we refer the reader to the books by Peleg~\cite{Peleg00} and Narasimhan and Smid~\cite{ns-gsn-07}.

\begin{table}[htpb]
\begin{center}
\scriptsize{
\begin{tabular}{|c|c|c|c|} \hline
  & Graph $\Theta_k$ & \raisebox{-2pt}{Graph $Y_k$} & Graph $YY_k$ \\[2pt]
\cline{2-4}
Parameter $k$ & \multicolumn{3}{c|}{\raisebox{-2pt}{$t$-Spanner, for constant real value $t \ge 1$?}} \\[2pt]
\hline\hline
\raisebox{-2pt}{$k \in \{2, 3\}$} & \multicolumn{3}{c|}{\raisebox{-2pt}{NO~\cite{MollaThesis09}}} \\[2pt]
\hline
\raisebox{-2pt}{$k = 4$} & \raisebox{-2pt}{\open} &
\raisebox{-2pt}{$t = 8\sqrt{2}(26+23\sqrt{2})$~\cite{BDD+10}} &
\raisebox{-2pt}{NO~\cite{DMP09}} \\[2pt]
\hline
\raisebox{-2pt}{$k = 5$} & \multicolumn{3}{c|}{\raisebox{-2pt}{\open}} \\[2pt]
\hline
\raisebox{-2pt}{$k = 6$} & \raisebox{-2pt}{$t = 2$~\cite{Bon+10}} & \raisebox{-2pt}{$t=17.64$~\cite{DR10}} &
\raisebox{-2pt}{NO~\cite{MollaThesis09}} \\[2pt]
\hline
\raisebox{-6pt}{$k > 6$} & \multirow{2}{*}{\raisebox{-1em}{$t = \frac{1}{1-2\sin(\pi/k)}$}} &
\raisebox{-6pt}{$t = \frac{1+\sqrt{2-2\cos(2\pi/k)}}{2\cos(2\pi/k) - 1}$~\cite{BDD+10}} &
\multirow{2}{*}{\raisebox{-1em}{\normalsize{\red{$k = 6k', k' \ge 6, t = 11.67$~~[this paper]}}}} \\[6pt]
\cline{1-1}\cline{3-3}
\raisebox{-3pt}{$k > 8$} &  &
\raisebox{-3pt}{$t = \frac{1}{\cos(2\pi/k) - \sin(2\pi/k)}$~\cite{bmnsz-agbsp-03}} & \\[6pt]
\hline\hline
\end{tabular}
} 
\vspace{-1em}
\end{center}
\caption{Spanning properties of Theta and Yao-based graphs.}
\label{tab:yao}
\end{table}

In this paper we take a first step towards proving that $YY_k$ is a spanner, for sufficiently large $k$.
Our main result is that $YY_{6k}$ is a $t$-spanner, for any $k \ge 6$ and $t = 11.67$. As far as we know, this is the first positive result regarding the spanning property of Sparse-Yao graphs. This result relies on a recent result by Bonichon et al.~\cite{Bon+10}, who prove that $\Theta_6$ is a $2$-spanner. Our main contribution is showing that  $YY_{6k}$ contains a short path between the endpoints of each edge in $\Theta_6$. More precisely, we show that
corresponding to each edge $ab \in \Theta_6$, there is a path between $a$ and $b$ in $YY_{6k}$ no longer than $t|ab|$, for $t = 5.832$ and $k \ge 6$. Combined with the fact that $\Theta_6$ is a $2$-spanner, this yields an upper bound of $11.67$ on the stretch factor of $YY_{6k}$. This result also shows that the class of Sparse-Yao spanners is infinite.

\subsection{Notation and Definitions}
Throughout the rest of the paper we work with the graphs $\Theta_6$, $Y_{6k}$ and $YY_{6k}$ defined for a fixed point set $\Pt$ and for positive integer $k \ge 2$. We view paths in these graphs as \emph{undirected}, and refer to the direction of an edge only when necessary to establish certain graph properties.
All three graphs $\Theta_6$, $Y_{6k}$ and $YY_{6k}$ use a first ray emanating from the origin of the coordinate system in the direction of the positive $x$-axis; each successive ray is obtained by a counter-clockwise rotation of the previous ray by angle $\alpha$ about the origin ($\alpha = 2\pi/6$ in the case of $\Theta_6$, and $\alpha = 2\pi/(6k)$ in the case of $Y_{6k}$ and $YY_{6k}$). A \emph{cone} is a region between two successive rays. Starting from the positive $x$-axis, the cones encountered in counter-clockwise order are $\cone_{\Theta1}, \cone_{\Theta2}, \ldots, \cone_{\Theta6}$ in the case of $\Theta_6$ (see Fig.~\ref{fig:thetacross}a), and $\cone_{Y1}(a), \cone_{Y2}(a), \ldots, \cone_{Y6k}(a)$ in the case of $Y_{6k}$. Note that the subscripts $_\Theta$ and $_Y$ are used to differentiate between the cones used in constructing $\Theta_6$ and those used in constructing $Y_{6k}$ and $YY_{6k}$; and the numerical subscripts are used to identify a particular cone from among all cones with the same apex.
Each cone $\cone$ is \emph{half-open} and \emph{half-closed} in the sense that it includes the ray clockwise from $\cone$ bounding $\cone$, but excludes the ray counter-clockwise from $\cone$ bounding $\cone$.

For any point $a \in \Pt$ and fixed cone $\cone$, let $\cone(a)$ denote the copy of $\cone$ translated so that its apex coincides with $a$.
For any two points $a, b \in \Pt$, we use $\cone_\Theta(a, b)$ to refer to the cone with apex $a$ that contains $b$, used in constructing $\Theta_6$.
We define $T(a, b)$ to be the open equilateral triangle with two of its sides along the
bounding rays for $\cone_\Theta(a, b)$, and the third side passing through $b$ (see, for example, the large shaded
triangle from Fig.~\ref{fig:proof1}a).

We say that two edges \emph{intersect} each other if they share a common point. If the common point is not an endpoint, the edges \emph{cross} each other. Fig.~\ref{fig:thetacross}b shows a pair of crossing edges; compare it to the two pairs of intersecting but non-crossing edges from Fig.~\ref{fig:thetacross}c. Throughout the paper, $\oplus$ is used to denote the
concatenation operator. A path in a graph between two points $a$ and $b$ is denoted by $p(a, b)$. To avoid confusion, we attach to the path notation one of the subscripts $_\Theta$, $_Y$ and $_{YY}$, depending on whether the path is in $\Theta_6$, $Y_{6k}$ or $YY_{6k}$. For example, $p_Y(a, b)$ refers to a path in $Y_6$ from $a$ to $b$.

The rest of the paper is organized as follows. Sec.~\ref{sec:basic} introduces a few isolated lemmas that are used in our main proof. The proofs of these lemmas are rather involved, and for this reason we defer them until Sec.~\ref{sec:proofs}, by which point their use is the main proof should be clearly understood. Sec.~\ref{sec:main} presents our main result. We wrap up with some conclusions and future work in Sec.~\ref{sec:conclusions}.

\section{Preliminaries}
\label{sec:basic}
In this section we provide a few isolated lemmas that will be used in the main proof. We defer the proofs of these lemmas until after the proof of the main theorem (Thm.~\ref{thm:maintheta}), so that the flow of ideas can be followed without interruption. We encourage the reader to skip ahead to \S\ref{sec:main}, and refer back to these lemmas from the context of Thm.~\ref{thm:maintheta}, where their usefulness will become evident.

\medskip
\noindent
We begin this section with the statement of a result established in~\cite{Bon+10}.

\begin{theorem}
For any pair of points $a, b\in \Pt$, there is a path in $\Theta_6$
whose total length is bounded above by $2|ab|$.
~\emph{\cite{Bon+10}}
\label{thm:theta6}
\end{theorem}
The bound $2$ on the stretch factor of $\Theta_6$ is tight~\cite{Bon+10}. The key ingredient in the result of Thm.~\ref{thm:theta6} is a specific subgraph of $\Theta_6$, called \emph{half-}$\Theta_6$. This graph preserves only half of the edges of $\Theta_6$, those belonging to non
consecutive cones. Bonichon et al.~\cite{Bon+10} show that half-$\Theta_6$ is a
\emph{triangular-distance}\footnote{The \emph{triangular distance} from a point $a$ to a point $b$
is the side length of the smallest equilateral triangle centered at $a$ and touching $b$.}
\emph{Delaunay triangulation}, computed as the dual of the Voronoi diagram based on
the triangular distance function.
This result, combined with Chew's result from~\cite{Chew89}, showing that any triangular-distance Delaunay triangulation is a $2$-spanner, yields the result of Thm.~\ref{thm:theta6}. For details, we refer the reader
to~\cite{Bon+10}.

\begin{figure}[htpb]
\centering
\includegraphics[width=0.5\linewidth]{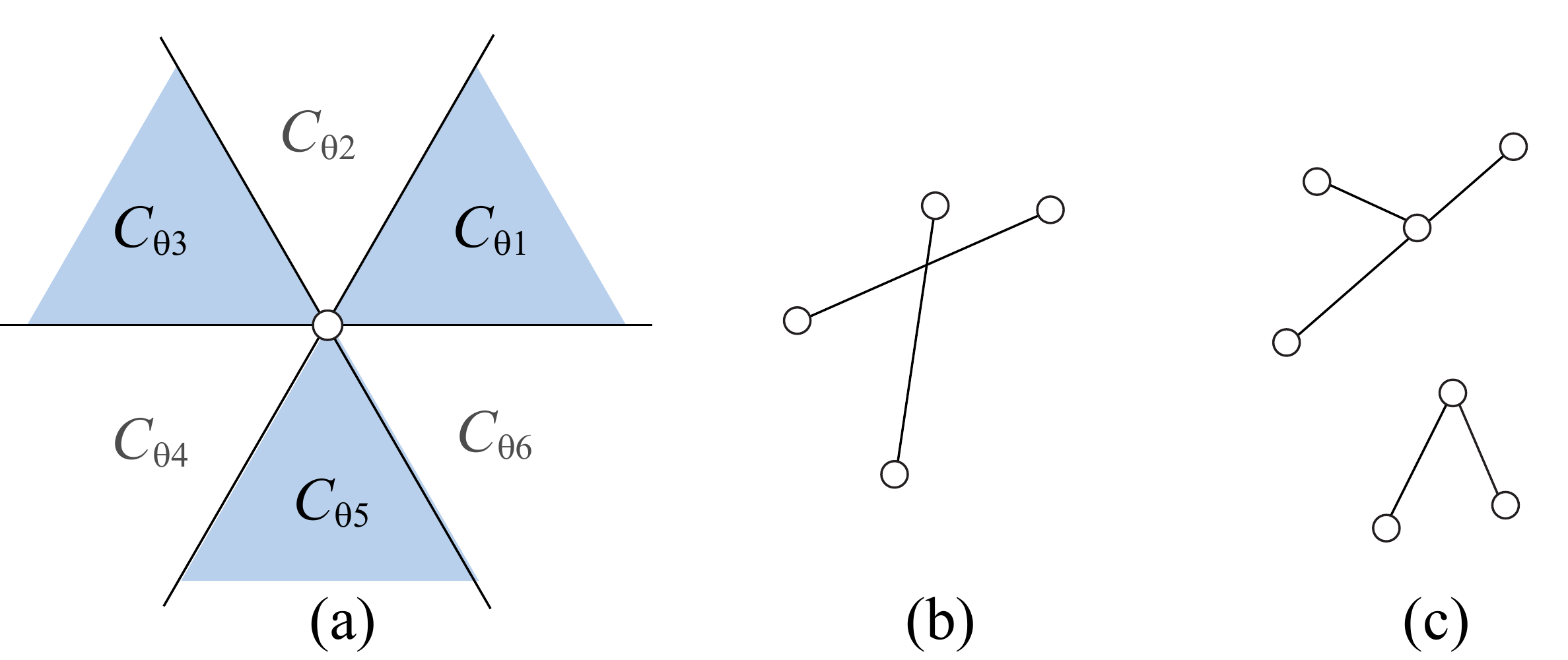}
\caption{(a) Cones numbering (b) Crossing edges (c) Intersecting, non-crossing edge pairs.}
\label{fig:thetacross}
\end{figure}

\noindent
Lem.~\ref{lem:thetapath} below will play a central role in the proofs of Lemmas~\ref{lem:thetapathbb} through~\ref{lem:thetapathaa3}.
\begin{lemma}
Let $a, b \in \Pt$ and let $x$ and $z$ be the other two vertices of $T(a, b)$. Let $y$ be the point on $az$ such that $by$ is parallel to $ax$. Let $p_\Theta(a, b)$ be a shortest path in $\Theta_6$ from $a$ to $b$. If $\triangle byz$ is empty of points in $\Pt$, then $|p_\Theta(a, b)| \le |ax| + |ay|$. Moreover, each edge of $p_\Theta(a, b)$ is no longer than $|ax|$. \emph{[Refer to Fig.~\ref{fig:S-trapezoid}a.]}
\label{lem:thetapath}
\end{lemma}
Note that Lem.~\ref{lem:thetapath} does not specify which of the two sides $ax$ and $az$ lies clockwise from $T(a,b)$, so the lemma applies in both situations.

\begin{figure}[htpb]
\centering
\includegraphics[width=0.9\linewidth]{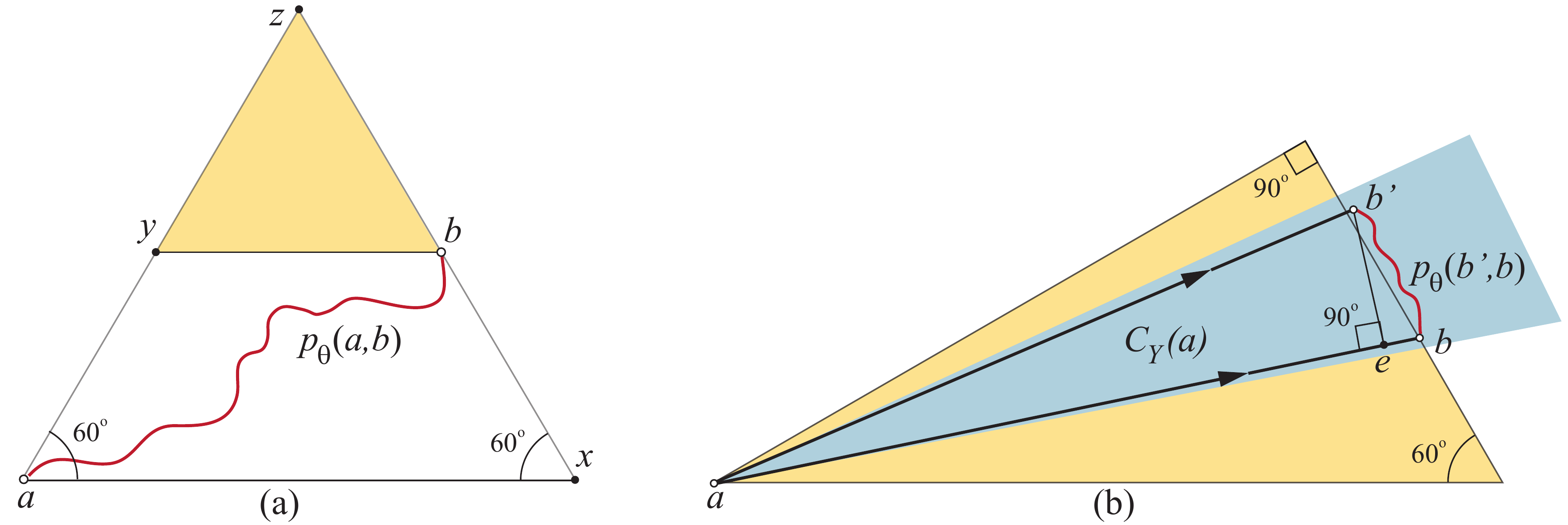}
\caption{(a) Lem.~\ref{lem:thetapath}: $|p_\Theta(a, b)| \le |ax| + |ay|$
(b) Lem.~\ref{lem:thetapathbb}: $ab \in \Theta_6$, $ab' \in Y_{6k}$, $|p_\Theta(b',b)| < \sqrt{3}|b'e|$. }
\label{fig:S-trapezoid}
\end{figure}

\medskip
\noindent
Lemmas~\ref{lem:thetapathbb} through~\ref{lem:thetapathaa3} isolate specific situations that will arise in the analysis of our main result. These lemmas can be stated and investigated independently.

\begin{lemma}
Fix an integer $k > 1$ and let $a$, $b$ and $b'$ be distinct points in $\Pt$ such that $\arr{ab} \in \Theta_6$ lies in $\cone_{\Theta1}(a)$ below the bisector of $\cone_{\Theta1}(a)$,
and $\arr{ab'} \in Y_{6k}$ lies in $\cone_Y(a, b)$. Let $e$ be the point on $ab$ such that $b'e$ is perpendicular on $ab$. Then there is a path $p_\Theta(b', b)$ in $\Theta_6$ of length
\[
|p_\Theta(b',b)| < \sqrt{3}|b'e|
\]
Furthermore, each edge of $p_\Theta(b',b)$ is strictly shorter than $ab$. \emph{[Refer to Fig.~\ref{fig:S-trapezoid}b.]}
\label{lem:thetapathbb}
\end{lemma}

\begin{lemma}
Fix an integer $k > 1$ and angle $\alpha = \frac{\pi}{3k}$. Let $a$, $a'$ and $b'$ be distinct points in $\Pt$ such that (i) $\arr{ab'} \in Y_{6k}$ lies in $\cone_{\Theta1}(a)$,
and (ii) $\arr{a'b'} \in YY_{6k}$ lies in $\cone_Y(b', a)$, such that $a' \in \cone_{\Theta3}(a) \cup \cone_{\Theta5}(a)$. Then there is a path $p_\Theta(a, a')$ in $\Theta_6$ of length
\[
|p_\Theta(a,a')| <  \frac{4}{\sqrt{3}}|a'b'|\sin(\alpha)
\mbox{~~~~~~~\emph{[Refer to Fig.~\ref{fig:thetapathaa1}.]}}
\]
\label{lem:thetapathaa1}
\end{lemma}
\vspace{-1.5em}
Intuitively, the distance $|aa'|$ is the context of Lem.~\ref{lem:thetapathaa1} is fairly small. For this reason,
Lem.~\ref{lem:thetapathaa1} claims an upper bound on the length of $p_\Theta(a,a')$ that is good enough for our purposes (in the sense that it is superseded by the upper bounds derived in the companion lemmas), but is not necessarily tight.

\begin{figure}[htpb]
\centering
\includegraphics[width=0.5\linewidth]{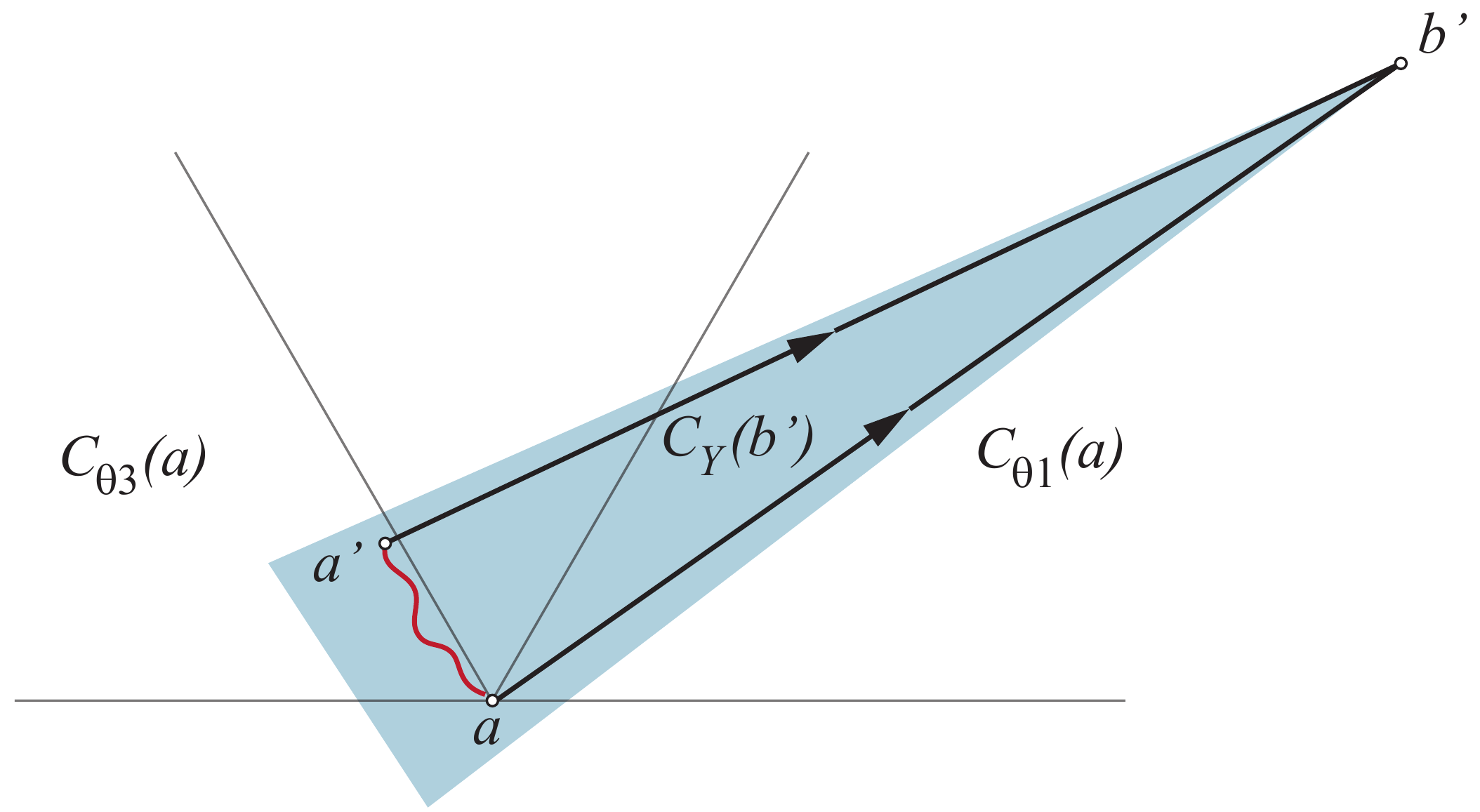}
\caption{Lem.~\ref{lem:thetapathaa1}: $ab' \in Y_{6k}$, $a'b' \in YY_{6k}$, $a' \in \cone_{\Theta3}(a)$, $|p_\Theta(a,a')| <  4|a'b'|\sin(\alpha)/\sqrt{3}$.}
\label{fig:thetapathaa1}
\end{figure}

\vspace{-1em}
\begin{lemma}
Fix an integer $k > 1$ and angle $\alpha = \frac{\pi}{3k}$. Let $a$, $b$, $a'$ and $b'$ be distinct points in $\Pt$ that satisfy the following properties: (i) $\arr{ab} \in \Theta_6$ lies in $\cone_{\Theta1}(a)$ below the bisector of $\cone_{\Theta1}(a)$, (ii) $\arr{ab'} \in Y_{6k}$ lies in $\cone_Y(a, b)$, and (iii) $\arr{a'b'} \in YY_{6k}$ lies in $\cone_Y(b', a)$ such that $a' \in \cone_{\Theta2}(a)$. Then there is a path $p_\Theta(a', a)$ in $\Theta_6$ of length
\[
|p_\Theta(a',a)| < |a'b'|\sin(\alpha)\left( 1 +
\max\left\{\sqrt{2}, \frac{2\sin(\frac{\pi}{6}+\alpha)}{\sqrt{3}\tan(\frac{\pi}{6}-\alpha)}\right\}\right).\]
\noindent
Furthermore, each edge of $p_\Theta(a',a)$ is strictly shorter than $ab$.
\emph{[Refer to Fig.~\ref{fig:S-thetapathaa2}a])}
\label{lem:thetapathaa2}
\end{lemma}

\begin{figure}[htpb]
\centering
\includegraphics[width=\linewidth]{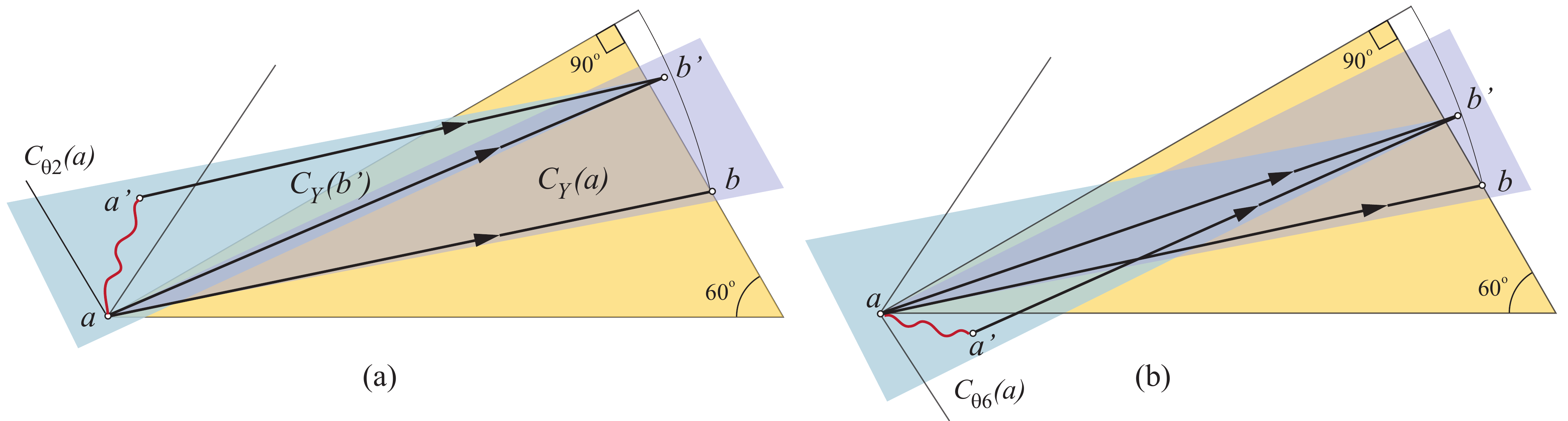}
\caption{$ab \in \Theta_6$, $ab' \in Y_{6k}$, $a'b' \in YY_{6k}$ (a) Lem.~\ref{lem:thetapathaa2}: $a' \in \cone_{\Theta2}(a)$ (b) Lem.~\ref{lem:thetapathaa3}: $a' \in \cone_{\Theta6}(a)$.}
\label{fig:S-thetapathaa2}
\end{figure}

\noindent
Lemma~\ref{lem:thetapathaa3} below complements Lem.~\ref{lem:thetapathaa2} regarding the relative position of $ab'$ and $a'b'$. The upper bound derived in this lemma may look somewhat unpolished, however it is intentionally left in a form that is most useful to the main theorem (Thm.~\ref{thm:maintheta}) from \S\ref{sec:main}.
\begin{lemma}
Fix an integer $k > 2$ and angle $\alpha = \frac{\pi}{3k}$. Let $a$, $b$, $a'$ and $b'$ be distinct points in $\Pt$ that satisfy the following properties: (i) $\arr{ab} \in \Theta_6$ lies in $\cone_{\Theta1}(a)$, below the bisector of $\cone_{\Theta1}(a)$ (ii) $\arr{ab'} \in Y_{6k}$ lies in $\cone_Y(a, b)$,
and (iii) $\arr{a'b'} \in YY_{6k}$ lies in $\cone_Y(b', a)$ below $ab'$ such that $a' \in \cone_{\Theta6}(a)$.
Let $h$ be the point on $ab$ such that $a'h$ is perpendicular on $ab$. Then there is a path $p_\Theta(a', a)$ in $\Theta_6$ of length
\[
|p_\Theta(a',a)| < |ah| + (1+2/\sqrt{3})|a'h|
\]
Furthermore, each edge of $p_\Theta(a',a)$ is strictly shorter than $ab$.
\emph{[Refer to Fig.~\ref{fig:S-thetapathaa2}b.]}
\label{lem:thetapathaa3}
\end{lemma}

\section{$YY_{6k}$ is a Spanner}
\label{sec:main}
This section contains our main result, which shows that there is an infinite class of sparse Yao graphs that are spanners. In particular, we show that $YY_{6k}$ is a $t$-spanner, for $k \ge 6$ and $t = 11.67$. Our approach takes advantage of the empty triangular area embedding each edge in $\Theta_6$, and establishes ``short'' paths in $YY_{6k}$ between the endpoints of each edge in $\Theta_6$. This, combined with the result of Thm.~\ref{thm:theta6}, yields our main result.

\begin{theorem}
For each edge $\arr{ab} \in \Theta_6$, there is a path $\pyy(a, b)$ in $YY_{6k}$ of length
$|\pyy(a, b)| \le t|ab|$, for any $k \ge 6$ and $t = 5.832$.
\label{thm:maintheta}
\end{theorem}
\begin{proof}
Let $\alpha = \pi/(3k)$. The proof is by induction on the length of the edges in $\Theta_6$. The base case corresponds to a shortest edge $\arr{ab} \in \Theta_6$.
In this case we show that $\arr{ab} \in Y_{6k}$ and $\arr{ab} \in YY_{6k}$. Assume to the contrary that
$\arr{ab} \not\in Y_{6k}$, and let $\arr{ac} \in Y_{6k}$, with $|ac| \le |ab|$, be the edge that lies in the cone $\cone$ with apex $a$ containing $ab$. The Law of Cosines applied on $\triangle abc$, along with the fact that $\ang{cab} < \alpha$, yields $|bc|^2 < 2|ab|^2(1-\cos(\alpha))$. Now note that for any $k \ge 3$, $\alpha \le \pi/9$ and $\sqrt{2(1-\cos(\alpha))} < 1/2$, therefore $|bc| < |ab|/2$. This along with Thm.~\ref{thm:theta6} shows that $\Theta_6$ contains a path $p_\Theta(b,c) \le 2|bc| < |ab|$. This implies that all edges on $p_\Theta(b,c)$ are strictly shorter than $ab$, contradicting our assumption that $ab$ is a shortest edge in $\Theta_6$. Similar arguments show that $ab \in YY_{6k}$, so the theorem holds for the base case.

For the inductive step, pick an arbitrary edge $ab \in \Theta_6$, and assume that the theorem holds for all edges in $\Theta_6$ strictly shorter than $ab$. We now seek a path $\pyy(a, b)$ in $YY_{6k}$ that satisfies the conditions of the theorem. We assume without loss of generality that $ab$ lies in the first cone $\cone_\Theta = \cone_{\Theta1}(a)$; if this is not the case, then we can always rotate the point set $\Pt$ about $a$ by a multiple of $\pi/3$ so that our assumption holds. (Note that the edge sets for $\Theta_6$ and $YY_{6k}$ remain unaltered by this rotation.) We can also assume that $ab$ lies along or below the bisector of $\cone_\Theta$; the situation in which $ab$ lies above the bisector of $\cone_\Theta$ is symmetric with respect to the bisector of $\cone_\Theta$.
Let $\cone_Y(a) \subset \cone_\Theta$ be the Yao cone with apex $a$ containing $b$.

We begin our analysis by considering the most complex case, in which a relevant edge from $\Theta_6$ is not in $Y_{6k}$, and a relevant edge from $Y_{6k}$ is not in $YY_{6k}$. We will see later that all other cases are particular instances of this complex case. Thus we start with the assumption that $ab \not\in Y_{6k}$ and hence $\cone_Y(a)$ must contain an edge $\arr{ab'} \in Y_{6k}$ with the property $|ab'| \le |ab|$. We proceed further through the complex case with the assumption that $\arr{ab'} \not\in YY_{6k}$. Let $\cone_Y(b')$ be the Yao cone with apex $b'$ containing $a$. Then $\cone_Y(b')$ must contain an edge $\arr{a'b'} \in YY_{6k}$ such that $|a'b'| \le |ab'|$. Because $ab \in \Theta_6$, $T(a, b)$ is empty of points in $\Pt$, therefore $b'$ must lie outside of $T(a,b)$. One immediate observation here is that $a'$ must also lie outside of $C_{\Theta1}(a)$; otherwise $|aa'| < |ab'|$, contradicting our assumption that $ab' \in Y_{6k}$.

We first determine a ``short'' path from $b'$ to $b$ in $YY_{6k}$. Let $e$ be the foot of the perpendicular from $b'$ on $ab$ (see Fig.~\ref{fig:proof1}a). By Lem.~\ref{lem:thetapathbb}, there is a path $p_\Theta(b',b)$ in $\Theta_6$ of length $|p_\Theta(b',b)| \le \sqrt{3}|b'e|$. Also according to Lem.~\ref{lem:thetapathbb}, each edge of $p_\Theta(b',b)$ is strictly shorter than $ab$. This enables us to apply the inductive hypothesis on each edge $xy \in p_\Theta(b',b)$, and claim the existence of a path $\pyy(x, y)$ in $YY_{6k}$ of length $|\pyy(x, y)| \le t|xy|$. Concatenating these paths and summing up the inequalities for all edges on $p_\Theta(b',b)$ yields a path $\pyy(b',b)$ in $YY_{6k}$ of length
\begin{equation}
|\pyy(b',b)| \le t\sqrt{3}|b'e|.
\label{eq:case0-1}
\end{equation}
We now express $b'e$ in terms of $ab'$ and angle $\alpha > \ang{b'ae}$ (this latter inequality holds because $b'$ and $e$ are in the same half-open Yao cone $\cone_Y(a)$ of angle $\alpha$). The fact $\alpha > \ang{b'ae}$ implies  $\sin(\alpha) > \sin(\ang{b'ae}) = |b'e| / |ab'|$, or in a simpler form $|b'e| < |ab'|\sin(\alpha)$. We substitute this inequality in~(\ref{eq:case0-1}) to obtain
\begin{equation}
|\pyy(b',b)| \le t|ab'|\sin(\alpha)\sqrt{3}.
\label{eq:case0-2}
\end{equation}
Next we focus our attention on determining a ``short'' path $\pyy(a, a')$ from $a$ to $a'$ in $YY_{6k}$. Given this, we can subsequently define a path $\pyy(a, b)$ from $a$ to $b$ as
\begin{equation}
\pyy(a,b) = \pyy(a, a') \oplus a'b' \oplus \pyy(b',b).
\label{eq:pab}
\end{equation}
Depending on the relative position of $a$ and $a'$, we must consider four possible cases. First note that $|a'b'| \le |ab'|$ implies that $\ang{a'ab'} < \pi/2$, and therefore $a'$ must lie in the open half-plane that contains $b'$ and is delimited by the perpendicular on $ab'$ through $a$. Because $ab' \in \cone_{\Theta1}(a)$, this half-plane shares no points with $C_{\Theta4}$, however it may share points with any other $\Theta$-cone apexed at $a$.
Thus $a' \in \cone_{\Theta2}(a) \cup \cone_{\Theta3}(a) \cup \cone_{\Theta5}(a) \cup \cone_{\Theta6}(a)$. We consider each of these situations in turn.

\begin{figure}[htpb]
\centering
\includegraphics[width=\linewidth]{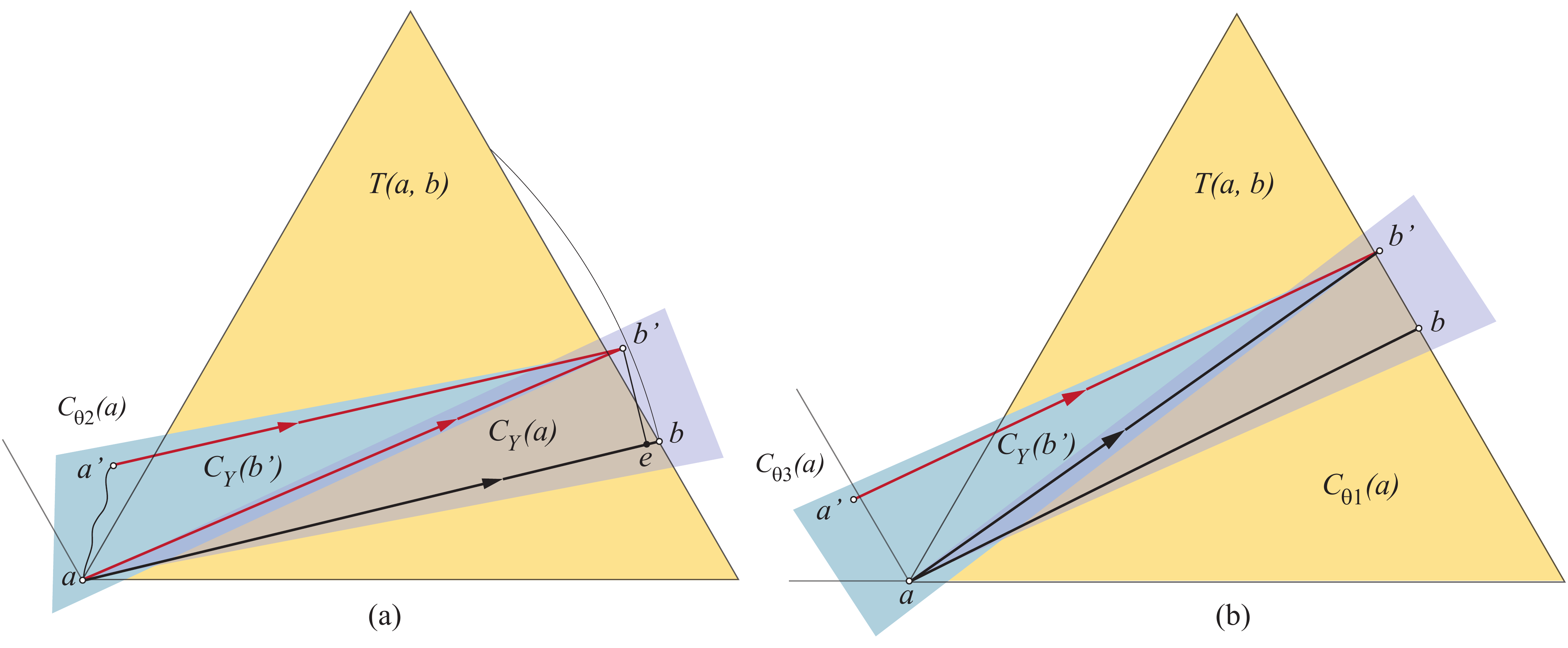}
\caption{Thm.~\ref{thm:maintheta} (a) Case 1: $a' \in \cone_{\Theta2}(a)$ (b) Case 2: $a' \in \cone_{\Theta3}(a)$.}
\label{fig:proof1}
\end{figure}

\paragraph{Case 1:} $a' \in \cone_{\Theta2}(a)$. (Refer to Fig.~\ref{fig:proof1}a.) This situation meets the conditions of Lem.~\ref{lem:thetapathaa2}, which tells us that $\Theta_{6}$ contains a path $p_\Theta(a',a)$ from $a'$ to $a$ of length
\begin{equation}
|p_\Theta(a',a)| < |a'b'|\sin(\alpha)(1 + M), \mbox{where }
M = \max\left\{\sqrt{2},\frac{2\sin(\pi/6+\alpha)}{\sqrt{3}\tan(\pi/6-\alpha)}\right\}.
\label{eq:case1-1}
\end{equation}
Lem.~\ref{lem:thetapathaa2} also tells us that each edge of $p_\Theta(a',a)$ is strictly shorter than $ab$. This enables us to use the inductive hypothesis and claim the existence of a path $\pyy(a',a)$ in $YY_{6k}$ of length
\begin{equation}
|p_\Theta(a',a)| < t|a'b'|\sin(\alpha)(1+M).
\label{eq:case1-2}
\end{equation}
Ignoring the direction of the edges, $\pyy(a,a') = \pyy(a',a)$ is a path in $YY_{6k}$ from $a$ to $a'$.
Substituting inequalities~(\ref{eq:case1-2}) and~(\ref{eq:case0-2}) in~(\ref{eq:pab}), and using the fact that
$|a'b'| \le |ab'| \le |ab|$, we derive an upper bound for the length of the path $\pyy(a, b)$ as
\begin{equation}
|\pyy(a,b)| < ab + t|ab|\sin(\alpha)(1 + \sqrt{3} + M).
\label{eq:case1-3}
\end{equation}
To prove the inductive step, we need to show that the right side of the inequality~(\ref{eq:case1-3}) does not exceed $t|ab|$, which (after eliminating the term $|ab|$) holds if
\begin{equation}
1 + t\sin(\alpha)(1+\sqrt{3}+ M) \le t.
\label{eq:case1-4}
\end{equation}
Two conditions must be met in order to satisfy inequality~(\ref{eq:case1-4}):
\[
  \left\{
  \begin{array}{l l}
   1-\sin(\alpha)(1+\sqrt{3}+ M) & > 0 \\
   \left(1-\sin(\alpha)(1+\sqrt{3}+M)\right)^{-1} & \le t\\
  \end{array} \right.
\]
We note that the term $M$ defined in~(\ref{eq:case1-1}) decreases as $\alpha$ decreases, and consequently the term on the left hand side of the first inequality above increases as $\alpha$ decreases. This property helps in verifying that the two inequalities above hold for any $\alpha \le \pi/18$ ($k \ge 6$) and $t \ge 5.832$. We also note that smaller $\alpha$ values imply smaller $t$; for example, for $\alpha \le \pi/30$ ($k \ge 10$), the constraint on $t$ is $t \ge 1.802$.

\paragraph{Case 2:} $a' \in \cone_{\Theta3}(a)$. (Refer to Fig.~\ref{fig:proof1}b.) This situation meets the conditions of Lem.~\ref{lem:thetapathaa1}, which tells us that $\Theta_{6}$ contains a path $p_\Theta(a,a')$ from $a$ to $a'$ of length $|p_\Theta(a',a)| < 4|a'b'|\sin(\alpha)/\sqrt{3}$. Note that for the values of $k$ imposed by the lemma, $\alpha \le \pi/18$ and the term $4|a'b'|\sin(\alpha)/\sqrt{3} < |a'b'| \le |ab'| \le |ab|$. Because the entire path $p_\Theta(a, a')$ is strictly shorter than $ab$, each edge of $p_\Theta(a,a')$ is also strictly shorter than $ab$, so we can use the inductive hypothesis to claim the existence of a path $\pyy(a,a')$ from $a$ to $a'$ in $YY_6$ of length
\begin{equation}
|\pyy(a,a')| \le 4t|a'b'|\sin(\alpha)/\sqrt{3}.
\label{eq:case2-1}
\end{equation}
Substituting inequalities~(\ref{eq:case2-1}) and~(\ref{eq:case0-2}) in~(\ref{eq:pab}), along with $|a'b'| \le |ab'| \le |ab|$,  yields 
\begin{equation}
|\pyy(a,b)| < ab + t|ab|\sin(\alpha)(4/\sqrt{3} + \sqrt{3}).
\label{eq:case2-2}
\end{equation}
To prove the inductive step, we need to show that the right side of the inequality~(\ref{eq:case2-2}) does not exceed $t|ab|$, which (after eliminating the term $|ab|$) holds if
\begin{equation}
1 + 7t\sin(\alpha)/\sqrt{3} \le t.
\label{eq:case2-3}
\end{equation}
Two conditions must be met in order to satisfy inequality~(\ref{eq:case2-3}):
\[
  \left\{
  \begin{array}{l l}
   1-7\sin(\alpha)/\sqrt{3} & > 0 \\
   \left(1-7\sin(\alpha)/\sqrt{3}\right)^{-1} & \le t\\
  \end{array} \right.
\]
Again, note that the term on the left hand side of the first inequality increases as $\alpha$ decreases. This property helps in verifying that the two inequalities above hold for any $\alpha \le \pi/18$ ($k \ge 6$) and $t \ge 3.36$. The lower bound on $t$ drops down to $2.11$ for $k = 8$, and lowers to $1.73$ for $k = 10$.

\paragraph{Case 3:} $a' \in \cone_{\Theta5}(a)$. This case is depicted in Fig.~\ref{fig:proof2}a. The result of  Lem.~\ref{lem:thetapathaa1} and the arguments used for Case 2 above apply here as well, yielding the same lower bounds for $k$ and $t$ (and upper bound for $\alpha$) as as in Case 2.

\begin{figure}[htpb]
\centering
\includegraphics[width=\linewidth]{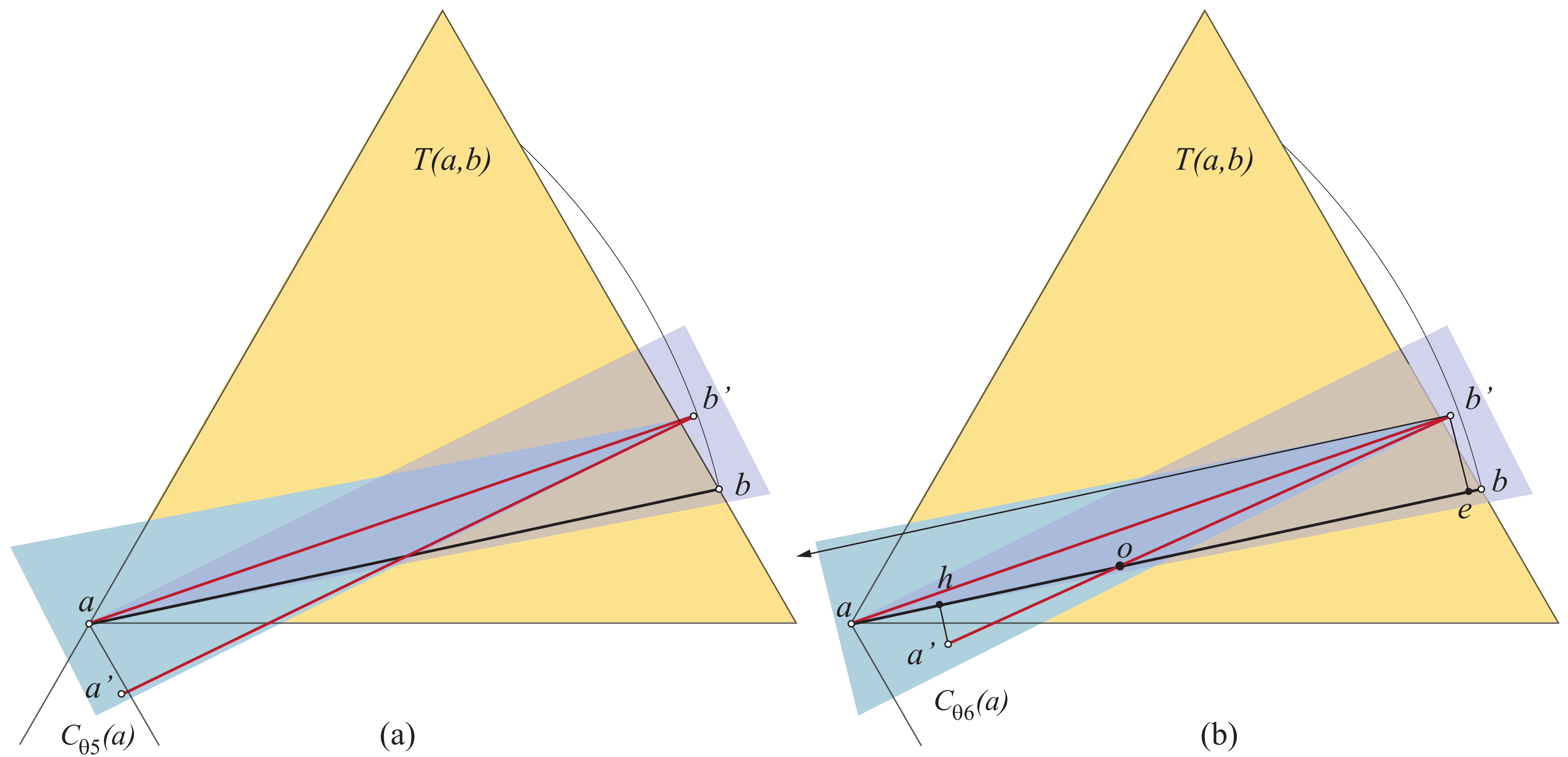}
\caption{Thm.~\ref{thm:maintheta} (a) Case 3: $a' \in \cone_{\Theta5}(a)$ (b) Case 4: $a' \in \cone_{\Theta6}(a)$.}
\label{fig:proof2}
\end{figure}

\paragraph{Case 4:} $a' \in \cone_{\Theta6}(a)$. This case is depicted in Fig.~\ref{fig:proof2}b.
Let $h$ be the foot of the perpendicular from $a'$ on $ab$. This context matches the one of Lem.~\ref{lem:thetapathaa3}, so we can use it to claim the existence of a path $p_\Theta(a',a)$ from $a'$ to $a$ of length $|p_\Theta(a',a)| < |ah| + (1+2/\sqrt{3})|a'h|$. Also according to Lem.~\ref{lem:thetapathaa3}, each edge of $p_\Theta(a',a)$ is strictly shorter than $ab$. This enables us to use the inductive hypothesis and claim the existence of a path $p_\Theta(a',a)$ in $YY_{6k}$ of length
\begin{equation}
|\pyy(a',a)| \le t|p_\Theta(a',a)| < t|ah| + t(1+2/\sqrt{3})|a'h|.
\label{eq:case4-1}
\end{equation}
Ignoring the direction of the edges, $\pyy(a, a') = \pyy(a', a)$ is a path from $a$ to $a'$ in $YY_{6k}$. By inequalities~(\ref{eq:case4-1}) and~(\ref{eq:case0-1}), an upper bound for the length of the path $\pyy(a, b)$ defined in~(\ref{eq:pab}) is
\begin{equation}
|\pyy(a,b)| < |a'b'| + t|ah| + t(1+2/\sqrt{3})(|a'h|+|b'e|).
\label{eq:mainpab1}
\end{equation}
(In deriving the right side term above, we used the fact that $\sqrt{3} < 1+2/\sqrt{3}$.)
Now note that the ray with origin $b'$ parallel to $ab$ lies inside $C_Y(b')$, therefore the angle formed by this ray with $a'b'$ is smaller than $\alpha$. It follows that $\ang{aoa'} < \alpha$. This in turn implies that
$\sin(\alpha) > \sin(\ang{aoa'}) = |a'h|/|oa'|$, or equivalently $|a'h| < |oa'|\sin(\alpha)$. Similarly,
$|b'e| < |ob'|\sin(\alpha)$. Substituting these two latter inequalities in~(\ref{eq:mainpab1}), and using
the fact that $|oa'| + |ob'| = |a'b'|$, yields the upper bound
\begin{equation}
|\pyy(a,b)| < |a'b'| + t|ah| + t(1+2/\sqrt{3})|a'b'|\sin(\alpha).
\label{eq:mainpab2}
\end{equation}
We now express $|ah| = |ae| - |he|$ in terms of $ab$ and $a'b'$. We have already established that
$\ang{aoa'} < \alpha$, therefore $\cos(\alpha) < \cos(\ang{aoa'}) = |ho|/|a'o|$, or equivalently $|ho| > |a'o|\cos(\alpha)$. Similarly,
$|oe| > |ob'|\cos(\alpha)$. Summing up these two inequalities yields $|he| = |ho|+|oe| > (|a'o|+|ob'|)\cos(\alpha) = |a'b'|\cos(\alpha)$. It follows that $|ah| = |ae| - |he| < |ab| - |a'b'|\cos(\alpha)$. Substituting this inequality in~(\ref{eq:mainpab2}) yields
\begin{equation}
|\pyy(a,b)| \le |a'b'| + t\left(ab - |a'b'|\cos(\alpha) + (1+2/\sqrt{3})|a'b'|\sin(\alpha)\right).
\label{eq:mainpab3}
\end{equation}
To prove the inductive step, we need to show that $|\pyy(a,b)| \le t|ab|$, which according to inequality~(\ref{eq:mainpab3}) holds if (after eliminating the term $|a'b'|$)
\[
1 - t\cos(\alpha) + t(1+2/\sqrt{3})\sin(\alpha) \le 0.
\]
Two conditions must be met in order to satisfy the inequality above:
\[
  \left\{
  \begin{array}{l l}
   \cos(\alpha) - (1+2/\sqrt{3})\sin(\alpha) & > 0 \\
   \left(\cos(\alpha) - (1+2/\sqrt{3})\sin(\alpha)\right)^{-1} & \le t\\
  \end{array} \right.
\]
We note that the term $\cos(\alpha) - (1+2/\sqrt{3})\sin(\alpha)$ increases as $\alpha$ decreases. This property helps in verifying that the two inequalities above hold for any $\alpha \le \pi/9$ ($k \ge 3$) and $t \ge 4.94$. The lower bound on $t$ drops down to $1.63$ for $k \ge 6$.

\medskip
\noindent
It remains to discuss the simpler cases in which $ab \in Y_{6k}$ or $ab' \in YY_{6k}$ (or both). We show that these are special cases of the above. Consider first the case in which $ab \in Y_{6k}$. If $ab \in YY_{6k}$ as well, then $p(a, b) = ab$ and the theorem holds. Otherwise, we let $b' = b$ and define $\pyy(a, b) = \pyy(a, a') \oplus a'b$; if $a' \in \cone_{\Theta6}(a)$ as in case 4 above, we also let $o = b$, so that the analysis for case 4 applies here as well; the other cases (1, 2 and 3) need no special adjustments. Similarly, if $ab' \in YY_{6k}$, we let $a'=a$ and define $\pyy(a, b) = ab' \oplus  \pyy(b', b)$; then the analysis for case 2 above settles this entire case.  Now note that the upper bound for $|\pyy(a,b)|$ yielded by the above analysis for these special cases is slightly smaller that the one obtained for the general case (because it does not include one of the strictly positive terms $|\pyy(a, a')|$ or $|\pyy(b', b)|$), therefore the spanning condition $\pyy(a, b) \le t|ab|$ holds for the same values of $t$ and $\alpha$.
\end{proof}

\medskip
\noindent
Thms.~\ref{thm:theta6} and~\ref{thm:maintheta} together yield the main result of this paper, stated in Thm.~\ref{thm:main} below.
\begin{theorem}
For any $k \ge 6$, $YY_{6k}$ is a $t$-spanner, with $t = 11.67$.
\label{thm:main}
\end{theorem}

\section{Proofs of Lemmas from \S\ref{sec:basic}}
\label{sec:proofs}

\subsection{Proof of Lem.~\ref{lem:thetapath}}
The proof is by induction on the pairwise distances between the points in $\Pt$.
\begin{figure}[htpb]
\centering
\includegraphics[width=0.85\linewidth]{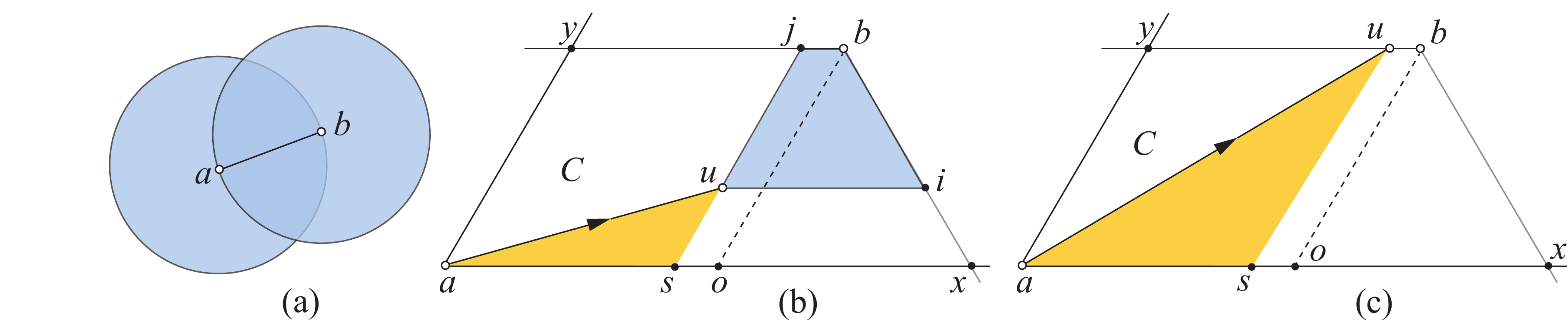}
\caption{Lemma~\ref{lem:thetapath}: (a) Open region empty of points in $\Pt$ (b,c) $au \in \Theta_6$ does not intersect $bo$.}
\label{fig:trapezoid1}
\end{figure}
The base case corresponds to a closest pair of points $a, b \in \Pt$. In this case, the circle centered at $a$ of radius $|ab|$ is empty of points in $\Pt$, and similarly for the circle centered at $b$ of radius $ab$. Any equilateral triangle with vertex $a$ and $b$ on its boundary fits inside the union of these two circles (see Fig.~\ref{fig:trapezoid1}a). This along with the fact that cones are half-open and half-closed implies that $\arr{ab} \in \Theta_6$ and the lemma holds for this case.

For the inductive case, pick an arbitrary pair of points $a, b \in \Pt$, and assume that the lemma holds for any pair of points at distance less than $|ab|$.
Let $\cone = \cone_\Theta(a, b)$, and let $\arr{au} \in \Theta_6$ be the edge in $\cone$ incident to $a$ (note that $\arr{au}$ exists, because $\cone$ contains $b$ and therefore is non-empty).
First note that $au$ may not cross over to the other side of $xz$, because in that case the projection of $u$ on the bisector of $\cone$ would be farther from $a$ than the projection of $b$ on the bisector of $\cone$, contradicting $\arr{au} \in \Theta_6$. This along with the fact that $\triangle ybz$ is empty of points in $\Pt$ shows that $u$ lies in the closed region $axby$. Next we focus on determining a ``short'' path $p_\Theta(u, b)$ from $u$ to $b$. Given this, we can subsequently define a path
\begin{equation}
p_\Theta(a, b) = au \oplus p_\Theta(u, b).
\label{eq:aub}
\end{equation}
Let $s$ be the intersection point between $ax$ and the line parallel to $ay$ passing through $u$. Let $o$ be the point on $ax$ such that $bo$ is parallel to $ay$. By the triangle inequality
\begin{equation}
|au| < |as| + |su|.
\label{eq:actri}
\end{equation}
We distinguish two cases, depending on whether $au$ intersects $bo$ or not.
Assume first that $au$ does not intersect $bo$. Then $|ub| < |ab|$. We have already established that $u$ lies in the closed region $axby$, so in this case $u$ lies either interior to $aoby$, or on $by$. Consider first the situation in which $u$ is interior to $aoby$, depicted in Fig.~\ref{fig:trapezoid1}b. In this case $b \in \cone(u)$. Let $i$ and $j$ be the intersection points between the bounding rays of $\cone(u)$ and $bx$ and $by$, respectively. Note that the equilateral triangle obtained by removing the trapezoid $uibj$ from $T(u, b)$ lies inside $\triangle ybz$, which is empty of points in $\Pt$ (by the lemma statement). Thus the inductive hypothesis applies here to show that there is a path $p_\Theta(u, b)$ from $u$ to $b$ of length
\begin{equation}
|p_\Theta(u, b)| \le |ui| +|uj|.
\label{eq:pcb}
\end{equation}
Substituting inequalities~(\ref{eq:pcb}) and~(\ref{eq:actri}) in~(\ref{eq:aub}) yields $|p_\Theta(a, b)| < (|as|+ |ui|) +(|su| + |uj|) \le |ax| + |ay|$. (Here we used the fact that $|ui| \le |sx|$ and $|su| + |uj| = |ay|$.) So the first claim of the lemma holds in this case. The second claim of the lemma follows immediately from the fact that $|au| < |ax|$ and the inductive hypothesis, by which each edge of $p_\Theta(u, b)$ is no longer than $|ui| \le |sx| < |ax|$.

If $u$ lies on the line segment $by$ (as in Fig.~\ref{fig:trapezoid1}c), then the trapezoid $uibj$ from Fig.~\ref{fig:trapezoid1}b degenerates to the line segment $ub$. The equilateral triangle with side $ub$ that lies in $\triangle byz$ coincides with one of $T(u, b)$ or $T(b, u)$, and is empty of points in $\Pt$. So the induction hypothesis applies again to show that $\Theta_6$ contains a path between $u$ and $b$ no longer than $|ub|$. This along with inequality~(\ref{eq:actri}) shows that the path $p_\Theta(a, b)$ is no longer than $|bu| + |us| + |sa| \le |ax| + |ay|$, so the lemma holds. (Note that the second claim of the lemma follows immediately from the fact that each of $au$ and $ub$ is no longer than $ax$.)

\begin{figure}[htpb]
\centering
\includegraphics[width=0.9\linewidth]{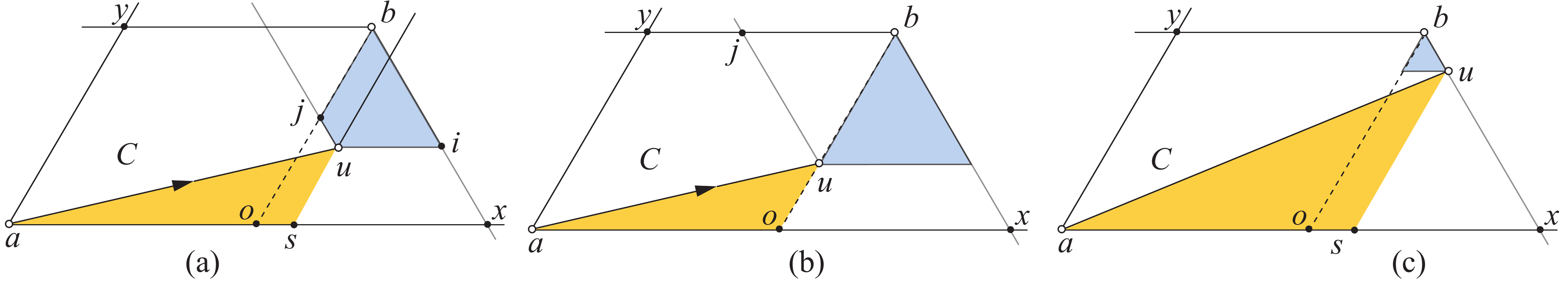}
\caption{Lemma~\ref{lem:thetapath}: $au \in \Theta_6$ intersects $bo$ (a) $u$ interior to $\triangle box$ (b) $u$ on $bo$ (c) $u$ on $bx$.}
\label{fig:trapezoid2}
\end{figure}

Assume now that $au$ intersects $bo$. Recall that $u$ must lie in the closed region $axby$, so in this case $u$ lies in the closed equilateral triangle $\triangle box$.
Consider first the situation in which $u$ lies strictly interior to $\triangle box$, as depicted in  Fig.~\ref{fig:trapezoid2}a. Let $i$ and $j$ be points on $bx$ and $bo$ respectively, such that $ui$ is parallel to $ax$ and $uj$ is parallel to $bx$. Note that the equilateral triangle obtained by removing $uibj$ from $T(b, u)$ is empty of points in $\Pt$, because it lies inside $T(a, u)$, which contains no points in $\Pt$.
This along with the fact that $|bu| < |ab|$ enables us to use the inductive hypothesis to claim the existence of a path $p_\Theta(b, u)$ from $b$ to $u$ of length
\begin{equation}
|p_\Theta(b, u)| < |bi| +|bj| = |bi| + |ui|.
\label{eq:pcb2}
\end{equation}
Because we seek undirected paths, we ignore the direction of the edges and let $p_\Theta(u, b) = p_\Theta(b, u)$. Substituting inequalities~(\ref{eq:pcb2}) and~(\ref{eq:actri}) in~(\ref{eq:aub}) yields
$|p_\Theta(a, b)| < (|as| +|ui|) + (|su| + |bi|) < |ax| + |bx|$. (Here we used the fact that $|bi| + |us| = |bi| + |ix| = |bx|$, and $|ui| \le |sx|$.) Also note that $|au| \le |ax|$, and each edge of $p_\Theta(b, u)$ is no longer than $|bi| \le |bx| < |ax|$ (by the inductive hypothesis). So the lemma holds for this case as well.

If $u$ lies on the line segment $bo$, then the equilateral triangle with side $bu$ lying inside $\triangle box$ coincides with one of $T(u, b)$ or $T(b, u)$. The first situation reduces to a special instance of the case depicted in Fig.~\ref{fig:trapezoid1}b, in which $j$ and $b$ coincide (so the trapezoid $biuj$ is really a triangle); the second situation reduces to a special instance of the case depicted in Fig.~\ref{fig:trapezoid2}a, in which $j$ and $u$ coincide. So the lemma holds for this case.

It remains to discuss the situation in which $u$ lies on the line segment $bx$, as depicted in Fig.~\ref{fig:trapezoid2}b. This situation occurs when $ab$ and $au$ are both candidates for the $\Theta_6$ edge selected in the cone $\cone$, and ties are broken in favor of $|au| \le |ab|$. Because $\arr{au} \in \Theta_6$, $T(a,u)$ is empty of points in $\Pt$. In particular, the equilateral triangle with side $ub$ that lies inside $T(a, b)$ (see the small shaded triangle in Fig.~\ref{fig:trapezoid2}b) is empty of points in $\Pt$. This triangle coincides with one of $T(u, b)$ or $T(b, u)$, so this is again a degenerate case in which the trapezoid $biuj$  reduces to the line segment $ub$. The inductive hypothesis applies here to show that $\Theta_6$ contains a path between $u$ and $b$ no longer than $|ub|$. This along with~(\ref{eq:actri}) and the fact that $|bu| + |us| = |bu| + |ux| = |bx|$ shows that the path $p_\Theta(a, b)$ is no longer than $|ax| + |bx|$. Also note that each of $au$ and $ub$ is no longer than $ax$, so the second claim of the lemma holds. This completes the proof.
\eproof

\subsection{Proof of Lem.~\ref{lem:thetapathbb}}

First observe that $|ab'| \le |ab|$, because $\arr{ab'} \in Y_{6k}$, and $b$ and $b'$ lie in the same Yao cone $\cone_Y(a)$. This implies that $e$ lies on the line segment $ab$ (otherwise $ab'$ would be longer than $ab$, a contradiction). Also note that $\arr{ab} \in \Theta_6$ implies that the interior of $T(a,b)$ is empty of points in $\Pt$, so $b'$ must lie to the right of $T(a, b)$ and above $b$. (This latter claim follows from the fact that
$\ang{abb'}$ is acute because $|ab'| < |ab|$. Also note that $b'$ cannot lie on the boundary of $T(a, b)$, because ties in $\Theta_6$ are broken in favor of the edge of shorter Euclidean length.) It can be easily verified that $b \in \cone_{\Theta5}(b')$.
\begin{figure}[htpb]
\centering
\includegraphics[width=0.45\linewidth]{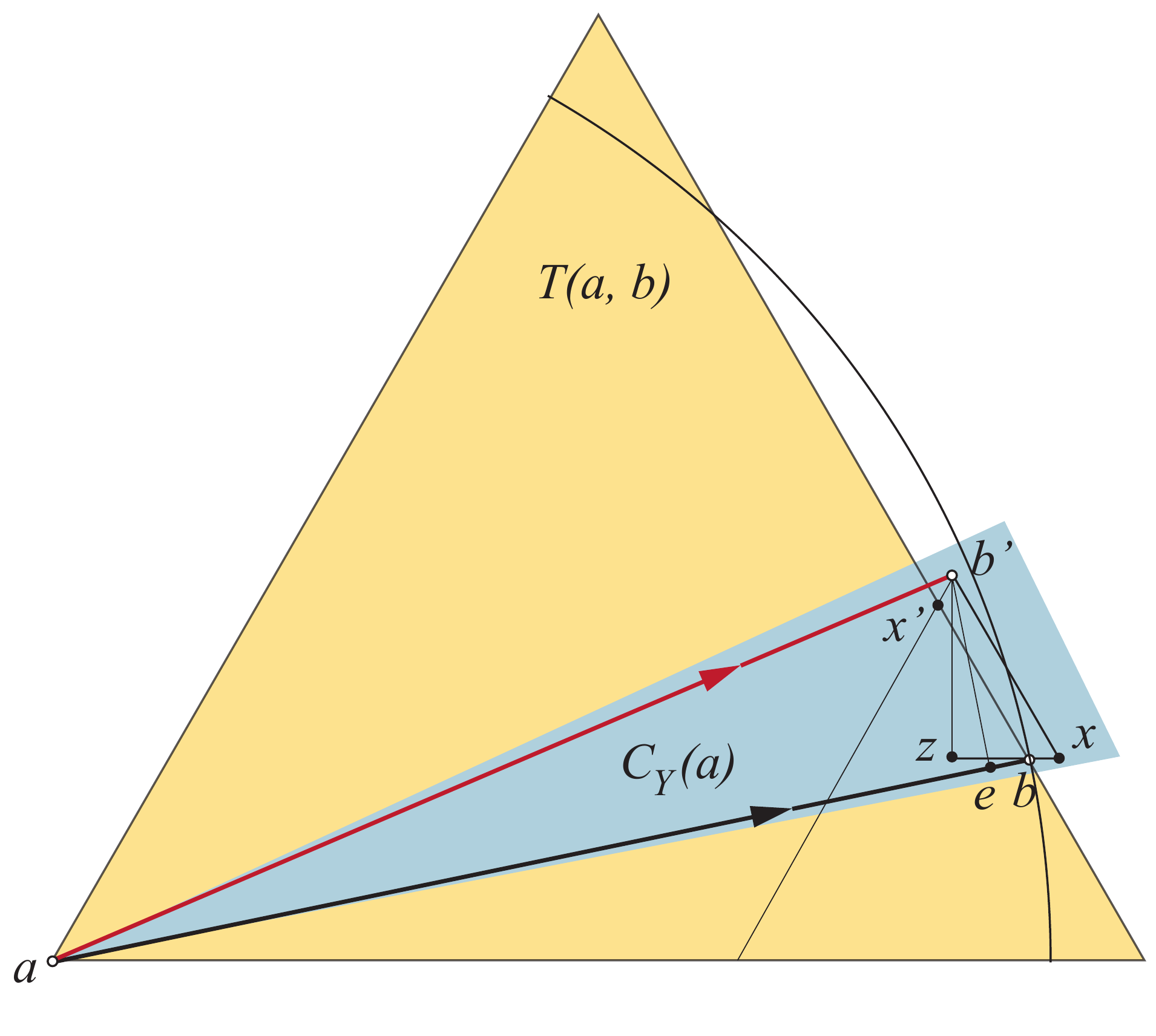}
\caption{Lem.~\ref{lem:thetapathbb}: $ab \in \Theta_6$, $ab' \in Y_{6k}$, $|p_\Theta(b,b')| \le \sqrt{3}|b'e|$.}
\label{fig:thetapathbb}
\end{figure}

Let $x$ be at the intersection between the line through $b'$ parallel to the right side of $T(a, b)$ and the horizontal line though $b$. Let $x'$ be at the intersection between the right side of $T(a, b)$ and the left ray of $\cone_{\Theta5}(b')$. Note that the equilateral triangle obtained after removing $bxb'x'$ from $T(b',b)$ lies inside $T(a, b)$ and therefore it is empty of points in $\Pt$.
%
%
This places us in the context of Lem.~\ref{lem:thetapath} so we can claim the existence of a path $p_\Theta(b',b)$ in $\Theta_6$ of length
\begin{equation}
|p_\Theta(b',b)| \le |b'x| + |b'x'| = |b'x| + |xb|
\label{eq:pbb0}
\end{equation}
Next we establish an upper bound on $|b'x|$ and $|xb|$ in terms of $|b'e|$. Let $z$ be the intersection point between the horizontal through $b$ and the vertical through $b'$. Then $\ang{zb'x} = \pi/6$, therefore $|b'x| = |b'z| / \cos(\pi/6) = 2|b'z|/\sqrt{3}$, and
$|xb| < |xz| = |b'z|\tan(\pi/6) = |b'z| / \sqrt{3}$. Summing up these two inequalities yields
$|b'x| + |xb| < |b'z|\sqrt{3}$. This along with the fact that $|b'z| \le |b'e|$ (the two terms are equal when $ab$ is horizontal) yields the upper bound stated by the lemma.

We now turn to the second claim of the lemma. By Lem.~\ref{lem:thetapath}, each edge of $|p_\Theta(b',b)|$ is no longer than $b'x$. Now note that $\ang{axb'} \ge \ang{bxb'} = \pi/3$ and $\ang{ab'x} > \ang{x'b'x} = \pi/3$. It follows that $\ang{b'ax} < \pi/3$. This along with the Law of Sines $|b'x|/\sin(\ang{b'ax}) = |ab'|/\sin(\ang{axb'})$ shows that $b'x$ is strictly shorter than $ab'$, which in turn is no longer than $ab$. This completes the proof.
\eproof

\subsection{Proof of Lem.~\ref{lem:thetapathaa1}}
Here we discuss only the case $a' \in \cone_{\Theta3}(a)$ depicted in Fig.~\ref{fig:thetapathaa1}; the case $a' \in \cone_{\Theta5}(a)$ is symmetric.
%
By Thm.~\ref{thm:theta6}, $\Theta_6$ contains a path $p_\Theta(a, a')$ no longer than $2|aa'|$. Thus we focus on bounding $|aa'|$.
First note that $\ang{ab'a'} < \alpha$ (because $ab'$ and $a'b'$ are in the same cone of angle $\alpha$), and $\ang{a'ab'} > \pi/3$ (because $\ang{a'ab'}$ includes the entire $\pi/3$-cone $\cone_{\Theta2}(a)$).
It follows that
$\sin(\ang{ab'a'}) < \sin(\alpha)$, and $\sin(\ang{a'ab'}) > \sin(\pi/3) = \sqrt{3}/2$. Substituting these inequalities in the Law of Sines applied on triangle $\triangle aa'b'$ yields
\[\frac{|aa'|}{\sin(\alpha)} < \frac{|aa'|}{\sin(\ang{ab'a'})} = \frac{|a'b'|}{\sin(\ang{a'ab'})} < \frac{2|a'b'|}{\sqrt{3}}.\]
This shows that $2|a'b'|\sin(\alpha)/\sqrt{3}$ is an upper bound for $|aa'|$, therefore $4|a'b'|\sin(\alpha)/\sqrt{3}$ is an upper bound for $|p_\Theta(a, a')| \le 2|aa'|$.
\eproof

\subsection{Proof of Lem.~\ref{lem:thetapathaa2}}
Recall that $\arr{ab} \in \Theta_6$ implies that $T(a, b)$ is empty of points in $\Pt$. This along with the fact that $a' \in \cone_{\Theta2}(a)$ implies that $a'b'$ crosses the left side of $T(a,b)$. Refer to Fig.~\ref{fig:thetapathaa2} throughout this proof.
\begin{figure}[htpb]
\centering
\includegraphics[width=\linewidth]{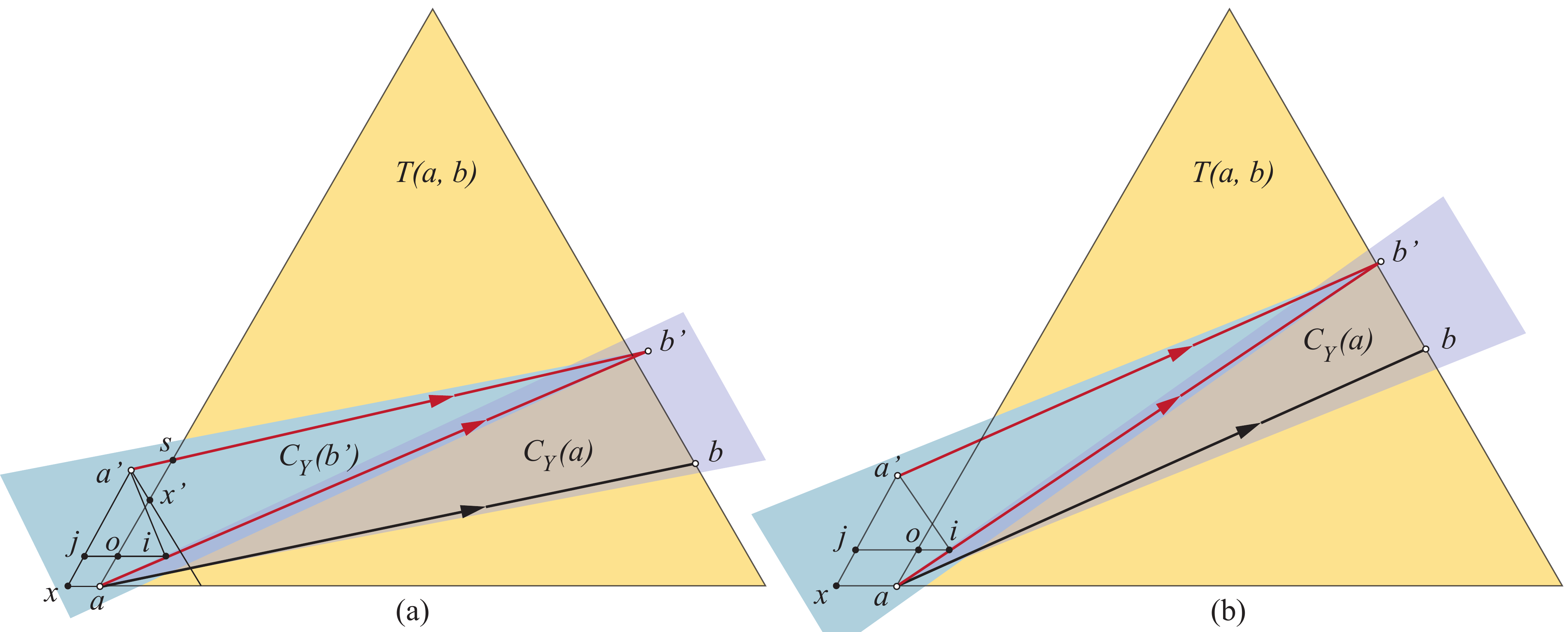}
\caption{Lem.~\ref{lem:thetapathaa2}: $ab \in \Theta_6$; $ab' \in Y_{6k}$; and $a'b' \in YY_{6k}$. (a) $ab'$ \emph{below} the bisector of $\cone_{\Theta1}(a)$ (b) $ab'$ \emph{above} the bisector of $\cone_{\Theta1}(a)$.}
\label{fig:thetapathaa2}
\end{figure}
Let $x$ be at the intersection between the line through $a'$ parallel to the left side of $T(a, b)$ and the horizontal line though $a$. By the lemma statement $a' \in \cone_{\Theta2}(a)$, therefore $a \in \cone_{\Theta5}(a')$. This implies that the right ray bounding $\cone_{\Theta5}(a')$ intersects the left side of $T(a, b)$; we call the intersection point $x'$.
(Note that if $a'$ lies on the boundary of $T(a, b)$, then the isosceles trapezoid $axa'x'$ degenerates to a line segment $aa'$; our arguments below apply to this scenario as well.)
The triangle obtained after removing $axa'x'$ from $T(a', a)$ lies inside $T(a, b)$, therefore it is empty of points in $\Pt$. This enables us to use
Lem.~\ref{lem:thetapath} and claim the existence of a path $p_\Theta(a',a)$ in $\Theta_6$ of length
\begin{equation}
|p_\Theta(a',a)| \le |a'x| + |a'x'| = |a'x| + |xa|.
\label{eq:paa1}
\end{equation}
Define three points $i$, $j$ and $o$ as follows: $i$ is the foot of the perpendicular from $a'$ on $ab'$; $j$ is the intersection point between the horizontal through $i$ and $a'x$; and $o$ is the intersection point between $ij$ and $ax'$. (If $a'$ is on the boundary of $T(a,b)$, then we let $a = x$, $a' = x'$ and $j = o$.) Note that $i$ must lie on the line segment $ab'$ -- otherwise, $a'b'$ would be longer than $ab'$, contradicting the fact that $\arr{a'b'} \in YY_{6k}$. We further expand the right side of~(\ref{eq:paa1}) into
\begin{equation}
|p_\Theta(a',a)| \le (|a'j| + |jx|) + |xa| = |a'j| + |oa| + |jo|.
\label{eq:paa1-1}
\end{equation}
We discuss two cases, depending on whether $ab'$ lies below or above the bisector of $\cone_{\Theta1}(a')$.

\paragraph{Case 1:} $ab'$ lies along or below the bisector of $\cone_{\Theta1}(a')$. (Refer to Fig.~\ref{fig:thetapathaa2}a). In this case, observe that
$a'i$ sits along or clockwise from $a'x'$ with respect to $a'$, because the line supporting $a'x'$ is orthogonal to the bisector of $C_{\Theta1}(a)$, and $ab'$ sits along or clockwise from the bisector with respect to $a$. Also note that $\ang{oia} \le \pi/6 \le \ang{oai}$. This along with the Law of Sines $|oa|/\sin(\ang{oia}) = |oi|/\sin(\ang{oai})$ implies $|oa| \le |oi|$. This enables us to further expand the term on the right side of inequality~(\ref{eq:paa1-1}) as follows:
\begin{equation}
|p_\Theta(a',a)| \le |a'j| + (|oi| + |jo|) = |a'j| + |ij|.
\label{eq:paa2}
\end{equation}
A similar analysis performed on triangle $\triangle a'ij$ shows that $\ang{ja'i} \le (\ang{a'ji} = \pi/3) \le \ang{a'ij}$, which along with the Law of Sines $|ij|/\sin(\ang{ja'i}) = |a'i|/\sin(\ang{a'ji})$ implies that
\begin{equation}
|ij| \le |a'i|.
\label{eq:paa-ij}
\end{equation}
Next we bound the term $|a'j|$ from inequality~(\ref{eq:paa2}) in terms of $|a'i|$ as well. Note that $a'i$, which is orthogonal to $ab'$, lies to the right of the vertical through $a'$, which in turn forms a $\pi/6$ angle with $a'j$. It follows that $\ang{ia'j} > \pi/6$ and $\ang{a'ij} < \pi/2$. This along with inequality~(\ref{eq:paa-ij}) and the Law of Cosines $|a'j|^2 = |a'i|^2 + |ij|^2 - 2|a'i||ij|\cos(\ang{a'ij})$
implies that
\begin{equation}
|a'j| < |a'i|\sqrt{2}
\label{eq:paa-a'j}
\end{equation}
Inequalities~(\ref{eq:paa2}),~(\ref{eq:paa-ij}) and~(\ref{eq:paa-a'j}) together yield
\begin{equation}
|p_\Theta(a',a)| < |a'i|(1+\sqrt{2})
\label{eq:paa3}
\end{equation}
Now note that $\ang{a'b'a} < \alpha$, which implies $\sin(\alpha) > \sin(\ang{a'b'a}) = |a'i| / |a'b'|$, or equivalently $|a'i| < |a'b'|\sin(\alpha)$. This along with inequality~(\ref{eq:paa3}) yields the upper bound
\begin{equation}
|p_\Theta(a',a)| < |a'b'|\sin(\alpha)(1+\sqrt{2}).
\label{eq:ub1}
\end{equation}
This upper bound matches the one claimed by the lemma when the $\max$ operator yields $\sqrt{2}$.

\paragraph{Case 2:} $ab'$ lies above the bisector of $\cone_{\Theta1}(a')$. (Refer to Fig.~\ref{fig:thetapathaa2}b;
we note that the relative position of $o$ and $x'$ on the upper ray of $C_{\Theta1}(a)$ is irrelevant to this case, so $x'$ is absent in Fig.~\ref{fig:thetapathaa2}b.)
Because $ab$ is below the bisector of $C_{\Theta1}(a)$ and $ab'$ is above the bisector, and because $\ang{bab'} < \alpha$, the inequalities $\pi/6 - \alpha \le \ang{oai} < \pi/6 < \ang{oia} \le \pi/6 + \alpha$ hold. Consequently,
$\sin(\pi/6 - \alpha) \le \sin(\ang{oai}) < \sin(\ang{oia}) \le \sin(\pi/6 + \alpha)$. This along with the Law of Sines $|oi|/\sin(\ang{oai}) = |oa|/\sin(\ang{oia})$ yields
\[
|oa| < |oi|\frac{\sin(\pi/6+\alpha)}{\sin(\pi/6-\alpha)}.
\]
We sum up $|oj|$ on both sides of the inequality above and substitute the result in~(\ref{eq:paa1-1}) to derive an upper bound
\begin{equation}
  |p_\Theta(a',a)| \le |a'j| + |oi|\frac{\sin(\pi/6+\alpha)}{\sin(\pi/6-\alpha)} + |oj| < |a'j| + |ij|\frac{\sin(\pi/6+\alpha)}{\sin(\pi/6-\alpha)}.
\label{eq:pcc2}
\end{equation}
(The latter inequality above uses the fact that $|oi| + |oj| = |ij|$ and $\sin(\pi/6+\alpha) > \sin(\pi/6-\alpha)$, for any $\alpha \le \pi/6$.)
Next we turn our attention to the triangle $\triangle a'ij$, to derive upper bounds for $|a'j|$ and $|ij|$. We have already established that $\pi/6 < \ang{oia} \le \pi/6 + \alpha$. This along with the fact that $\ang{a'ia} = \pi/2$ implies that $\pi/3 > \ang{a'ij} \ge \pi/3-\alpha$, which in turn implies that $\pi/3 < \ang{ja'i} \le \pi/3+\alpha$ (recall that $\ang{a'ji} = \pi/3$). This along with the Law of Sines
$|ij|/\sin(\ang{ja'i}) = |a'i|/\sin(\pi/3) = |a'j|/\sin(\ang{ja'i})$ implies that $|a'j| < |a'i|$ and
$|ij| < \frac{2}{\sqrt{3}}|a'i|\sin(\pi/3+\alpha) = \frac{2}{\sqrt{3}}|a'i|\cos(\pi/6-\alpha)$.
Substituting these inequalities in~(\ref{eq:pcc2}) yields
\begin{equation}
|p_\Theta(a',a)| < |a'i| + |a'i|\frac{2\sin(\pi/6+\alpha)}{\sqrt{3}\tan(\pi/6-\alpha)}.
\label{eq:pcc3}
\end{equation}
Because $\ang{a'b'a} < \alpha$, the inequality $|a'i| < |a'b'|\sin(\alpha)$ holds. Substituting this in~(\ref{eq:pcc3}), and combining the result with the upper bound from~(\ref{eq:ub1}) obtained for the first case,
yields the upper bound stated by the lemma.

We now turn to the second claim of the lemma. By Lem.~\ref{lem:thetapath}, no edge of $p_\Theta(a',a)$ exceeds $|a'x|$. Let $s$ be at the intersection between $a'b'$ and the upper ray of $C_{\Theta1}(a)$, and note that
$a'x$ is strictly shorter than $|as|$, which in turn is strictly shorter than $ab'$ (because $\ang{asb'}$ is obtuse). This, along with $|ab'| \le |ab|$, yields the second claim of the lemma.
\eproof

\subsection{Proof of Lem.~\ref{lem:thetapathaa3}}
Recall that $\arr{ab} \in \Theta_6$ implies that $T(a, b)$ is empty of points in $\Pt$. This along with the fact that $a'b'$ lies below $ab'$ implies that $a'b'$ crosses the bottom side of $T(a,b)$, as well as $ab$. Refer to Fig.~\ref{fig:thetapathaa3} throughout this proof.
\begin{figure}[htpb]
\centering
\includegraphics[width=0.5\linewidth]{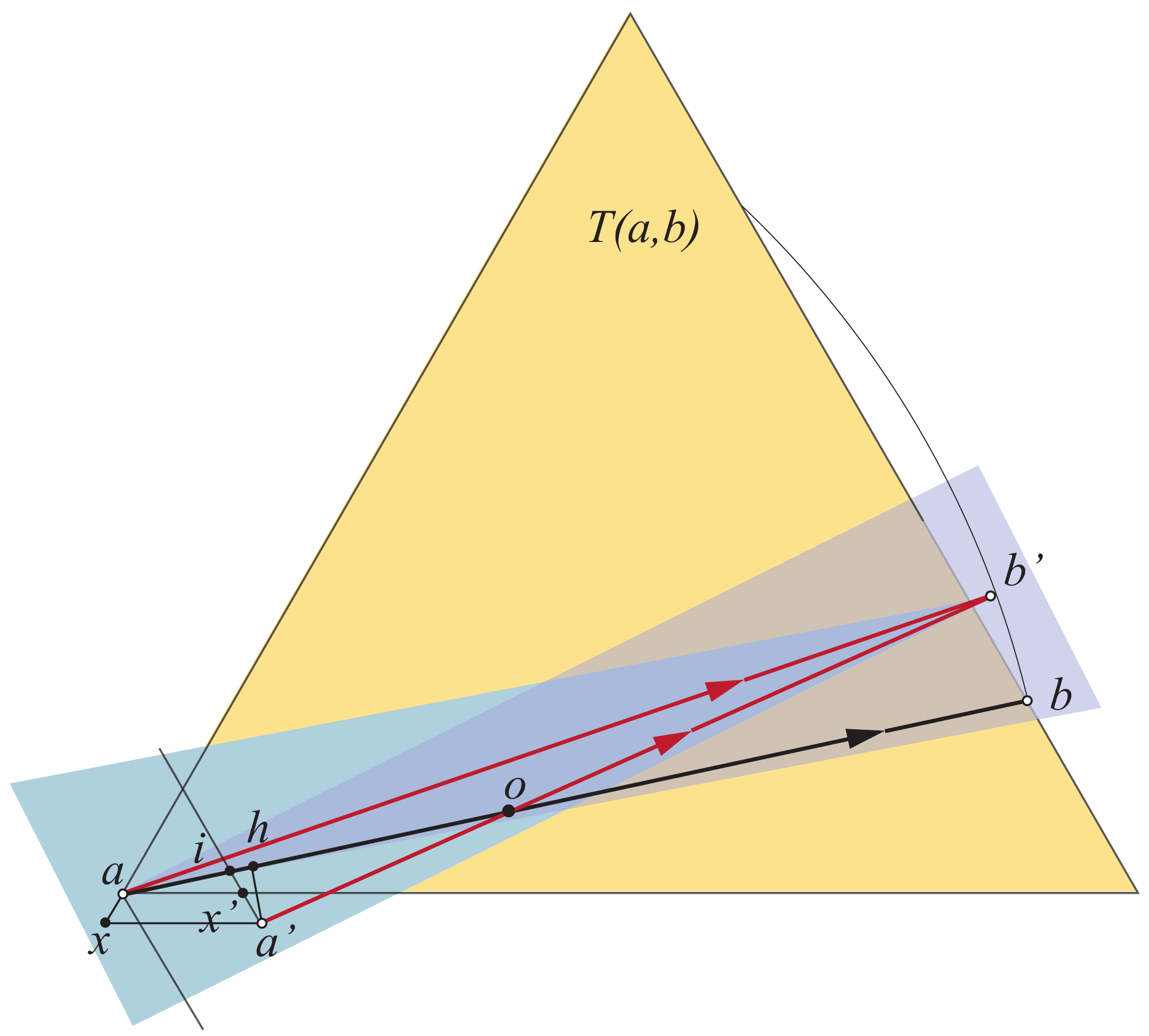}
\caption{Lem.~\ref{lem:thetapathaa3}: $ab \in \Theta_6$, $ab' \in Y_{6k}$, $a'b' \in YY_{6k}$, $a' \in \cone_{\Theta6}(a)$, $p_\Theta(a',a)$ is small.}
\label{fig:thetapathaa3}
\end{figure}
%
Because $|ab'| < |ab|$, $\ang{abb'}$ is acute, forcing $ab'$ to sit above $ab$. We now show that $h$ must lie on the line segment $ab$. To see this, consider the foot $f$ of the perpendicular from $a'$ on $ab'$ (not marked in Fig.~\ref{fig:thetapathaa3}, to avoid excessive labeling). The point $f$ must lie on the line segment $ab'$ --  otherwise $a'b'$ would be longer than $ab'$, contradicting $\arr{a'b'} \in YY_{6k}$. This, together with the fact that $a'h$ sits clockwise from  $a'f$ with respect to $a'$ (because $ab$ sits clockwise from $ab'$ with respect to $a$), shows that $h$ must lie to the right of $a$ as well.

Let $x$ be the intersection point between the horizontal through $a'$ and the left bounding ray of $\cone_{\Theta5}(a)$. Note that $a \in \cone_{\Theta3}(a')$, because $a' \in \cone_{\Theta6}(a)$ (by the lema statement). This implies that the upper bounding ray of $\cone_{\Theta3}(a')$ intersects the bottom side of $T(a,b)$; we label the intersection point $x'$. Note that the triangle obtained after removing $axa'x'$ from $T(a',a)$ lies inside $T(a, b)$ and therefore is empty of points in $\Pt$. This enables us to apply
Lem.~\ref{lem:thetapath} and claim the existence of a path $p_\Theta(a',a)$ in $\Theta_6$ of length
\begin{equation}
|p_\Theta(a',a)| \le |a'x| + |a'x'|
\label{eq:paa0-1}
\end{equation}
Let $i$ be the intersection point between $ab$ and the line supporting $a'x'$. (Note that $i$ lies to the left of $h$, because $a'i$ is orthogonal to the bisector of $C_{\Theta1}(a)$, which lies counterclockwise from $ab$ with respect to $a$.) Clearly $|a'x'| \le |a'i|$. Now observe that $\ang{ia'h} \le \pi/6$; this is because the angle formed by $ia'$ with the vertical through $a'$ is precisely $\pi/6$, and $a'h$ lies along or to the left of the vertical through $a'$. This implies that $\cos(\pi/6) \le \cos(\ang{ia'h}) = |a'h|/|a'i|$, or equivalently $|a'i| \le 2|a'h|/\sqrt{3}$. This along with the fact that $|a'x'| \le |a'i|$ implies
\begin{equation}
|a'x'| \le 2|a'h|/\sqrt{3}.
\label{eq:a'x'2}
\end{equation}
The triangle inequality applied on $\triangle aix'$ tells us that $|ax'| < |ai| + |ix'|$, which substituted in $|a'x| = |a'x'| + |x'a|$ yields $|a'x| < |ai| + |ia'|$. This along with the triangle inequality $|ia'| < |ih| + |a'h|$ yields $|a'x| < |ah| + |a'h|$. Summing up this latter inequality with~(\ref{eq:a'x'2}), and substituting the result in~(\ref{eq:paa0-1}), yields the upper bound stated by the lemma.

Let $o$ be the intersection point between $a'b'$ and $ab$. For the second claim of the lemma, we use the fact that no edge of $p_\Theta(a',a)$ is longer than $a'x$ (by Lem.~\ref{lem:thetapath}), which in turn is strictly shorter than $ao$ where $|ao| < |ab|$.
\eproof

\section{Conclusions}
\label{sec:conclusions}
This paper establishes the first positive result regarding the spanning property of Sparse-Yao graphs (also known as Yao-Yao graphs). We show that, for any $k \ge 6$, the Sparse-Yao graph $YY_{6k}$ is a spanner with stretch factor $11.67$; the stretch factor drops down to $4.75$ for $k \ge 8$. We leave open our conjecture that $YY_{k}$ has constant stretch factor for any $k$ larger than a specific threshold value (no less than $6$).

\medskip

\begin{thebibliography}{10}

\bibitem{Bon+10}
N.~Bonichon, C.~Gavoille, N.~Hanusse, and D.~Ilcinkas.
\newblock Connections between {T}heta-graphs, {D}elaunay triangulations, and
  orthogonal surfaces.
\newblock In {\em Proc. of the 36th international conference on Graph-theoretic
  concepts in computer science}, WG'10, pages 266--278, Berlin, Heidelberg,
  2010. Springer-Verlag.

\bibitem{BDD+10}
P.~Bose, M.~Damian, K.~Dou\"{\i}eb, J.~O'Rourke, B.~Seamone, M.~Smid, and
  S.~Wuhrer.
\newblock Pi/2-angle {Y}ao graphs are spanners.
\newblock {\em CoRR}, abs/1001.2913, 2010.

\bibitem{bmnsz-agbsp-03}
P.~Bose, A.~Maheshwari, G.~Narasimhan, M.~Smid, and N.~Zeh.
\newblock Approximating geometric bottleneck shortest paths.
\newblock {\em Computational Geometry: Theory and Applications},
  29(3):233--249, 2004.

\bibitem{Chew89}
L.~Paul Chew.
\newblock There are planar graphs almost as good as the complete graph.
\newblock {\em Journal of Computer and System Sciences}, 39:205--219, 1989.

\bibitem{DMP09}
M.~Damian, N.~Molla, and V.~Pinciu.
\newblock Spanner properties of $\pi/2$-angle {Y}ao graphs.
\newblock In {\em Proc. of the 25th European Workshop on Computational
  Geometry}, pages 21--24, March 2009.

\bibitem{DR10}
M.~Damian and K.~Raudonis.
\newblock Yao graphs span {T}heta graphs.
\newblock In {\em Proc. of the 4th International Conference on Combinatorial
  Optimization and Applications - Volume Part II}, COCOA'10, pages 181--194,
  Berlin, Heidelberg, December 2010. Springer-Verlag.

\bibitem{DegreeMac06}
B.~Hamdaoui and P.~Ramanathan.
\newblock Energy-efficient and {MAC}-aware routing for data aggregation in
  sensor networks.
\newblock In {\em Sensor Network Operations}, chapter 5.3, pages 291--308.
  Wiley-IEEE Press, March 2006.

\bibitem{MollaThesis09}
N.~Molla.
\newblock Yao spanners for wireless ad hoc networks.
\newblock Technical report, M.S. Thesis, Department of Computer Science,
  Villanova University, December 2009.

\bibitem{ns-gsn-07}
G.~Narasimhan and M.~Smid.
\newblock {\em Geometric Spanner Networks}.
\newblock Cambridge University Press, New York, NY, USA, 2007.

\bibitem{Peleg00}
D.~Peleg.
\newblock {\em Distributed computing: a locality-sensitive approach}.
\newblock Society for Industrial and Applied Mathematics, Philadelphia, PA,
  USA, 2000.

\end{thebibliography}

\def\cprime{$'$}

\end{document}